\documentclass[11pt, reqno]{amsart}

\usepackage{amssymb,amsthm,amsmath}
\usepackage{times}

\usepackage{algpseudocode}
\usepackage{algorithm}
\usepackage{float}

\usepackage{basicreq}

\newtheorem{theorem}{Theorem}
\newtheorem{definition}{Definition}
\newtheorem{lemma}[theorem]{Lemma}
\newtheorem{proposition}[theorem]{Proposition}
\newtheorem{corollary}{Corollary}

\renewcommand{\thefootnote}{\fnsymbol{footnote}}

\allowdisplaybreaks

\title[Phase Transitions in PML]{Phase Transitions for the Uniform Distribution \\ in the PML Problem and its Bethe Approximation$^\ast$}

\author{Chun Lam Chan} 
\address{Dept.\ Information Engineering, Chinese University of Hong Kong}
\email{clchan.eric@gmail.com}

\author{Winston Fernandes}
\address{Dell R\&D, Bangalore}
\email{winfernu@gmail.com}

\author{Navin Kashyap}
\address{Dept.\ Electrical Communication Engineering, Indian Institute of Science, Bangalore}
\email{nkashyap@ece.iisc.ernet.in}

\author{Manjunath Krishnapur}
\address{Dept.\ Mathematics, Indian Institute of Science, Bangalore}
\email{manju@math.iisc.ernet.in}

\begin{document}
\maketitle

\footnotetext[1]{This work was presented in part at the 2013 IEEE Information Theory Workshop held in Seville, Spain, Sept. 9--13, 2013.}

\begin{abstract}
The pattern maximum likelihood (PML) estimate, introduced by Orlitsky et al., is an estimate of the multiset of probabilities in an unknown probability distribution $\p$, the estimate being obtained from $n$ i.i.d.\ samples drawn from $\p$. The PML estimate involves solving a difficult optimization problem over the set of all probability mass functions (pmfs) of finite support. In this paper, we describe an interesting phase transition phenomenon in the PML estimate: at a certain sharp threshold, the uniform distribution goes from being a local maximum to being a local minimum for the optimization problem in the estimate. We go on to consider the question of whether a similar phase transition phenomenon also exists in the Bethe approximation of the PML estimate, the latter being an approximation method with origins in statistical physics. We show that the answer to this question is a qualified ``Yes''.  Our analysis involves the computation of the mean and variance of the $(i,j)$th entry, $a_{i,j}$, in a random $k \times k$ non-negative integer matrix $A$ with row and column sums all equal to $M$, drawn according to a distribution that assigns to $A$ a probability proportional to $\prod_{i,j} \frac{(M-a_{i,j})!}{a_{i,j}!}$.
\end{abstract}


\renewcommand{\thefootnote}{\arabic{footnote}}

\section{Introduction\label{sec:intro}}
Consider the estimation problem in which, given a sequence of $n$ i.i.d.\ samples from a fixed but unknown underlying probability distribution $\p$, we are required to estimate the multiset of probabilities in $\p$. In particular, we need not determine the correspondence between the symbols of the underlying alphabet and the probabilities in the multiset. Such a problem arises naturally in the context of universal compression of large-alphabet sources \cite{OSZ04}, and has several other applications, for example, population estimation from a small number of samples \cite{OSVZ11}. The multiset of empirical frequencies of the symbols observed in the $n$ samples is a straightforward estimate of the multiset of probabilities in $\p$; this estimate corresponds to the maximum likelihood (ML) estimate of $\p$. However, when the sample size, $n$, is smaller than the size of the support of the underlying distribution $\p$, the ML estimate may not give a good estimate of the multiset of probabilities in $\p$. An alternative estimate that has been proposed for this regime is the pattern maximum likelihood (PML) estimate, introduced by Orlitsky et al.~\cite{OSZ04}, \cite{OSVZ11} and described below.


The \emph{pattern} $\psibf$ or $\psibf(\x^n)$ of a sequence $\x^n = x_1,\ldots,x_n$ is a data structure that keeps track of the order of occurrence and the multiplicities of the distinct symbols in the sequence $\x^n$; for a precise definition, see Section~\ref{sec:pml}. The \emph{pattern maximum likelihood (PML) distribution} of a pattern $\psibf$ is the multiset of probabilities that maximizes the probability of observing a sequence with pattern $\psibf$. It has been argued \cite{OSVZ11}, \cite{OSVZ04}, \cite{OSSVZ04} that for a sequence $\x^n$ sampled from an unknown underlying probability distribution $\p$, the PML distribution of $\psibf(\x^n)$ is a good estimate of the multiset of probabilities in $\p$, even in situations where $n$ is much smaller than the support size of $\p$. However, for the purposes of this paper, we view the PML distribution purely as an interesting mathematical object.

The problem of determining the PML distribution (henceforth termed the ``PML problem'') of a given pattern $\psibf$ appears to be computationally hard \cite{OSSVZ04}--\cite{ADOP11}. In part, this is because the underlying optimization problem is not convex. It turns out that the PML problem can be very well approximated by its Bethe approximation \cite{Von12}, \cite{Fer13}, which in this case is a convex optimization problem. The Bethe approximation is a technique with roots in statistical physics. The optimization problem in the Bethe approximation can usually be solved highly efficiently using belief propagation algorithms \cite{YFW05}. 

In this paper, we are concerned with a remarkable phase transition phenomenon\footnote{Our use of the term ``phase transition'' here is inspired by statistical physics, where the term is often used to describe abrupt changes in behaviour of physical (especially, thermodynamical) systems.} observed in the PML problem. For a positive integer $k$, let $U_k$ denote the uniform distribution on $k$ symbols. Given a pattern $\psibf$, we can explicitly compute a quantity $\Upsilon(\psibf)$ such that for all $k < \Upsilon(\psibf)$, $U_k$ is a \emph{local maximum}, among all distributions $\p$ with support size $k$, for the optimization problem within the PML problem; and for all $k > \Upsilon(\psibf)$, $U_k$ is a \emph{local minimum}. On the basis of this observation, we proposed in \cite{FK13} a heuristic algorithm for determining whether or not the uniform distribution is the PML distribution of a given pattern $\psibf$.

Given that the Bethe approximation is a very good proxy for the PML distribution, it is natural to ask whether the phase transition phenomenon described above extends to the Bethe approximation as well. We are able to give a qualified affirmative answer to this question. Our answer is given in terms of a sequence of ``degree-$M$ optimization problems'' such that the degree-$1$ problem is the original PML problem, and as $M \to \infty$, we obtain the Bethe approximation. We show that for all sufficiently large $M$, the degree-$M$ optimization problems admit a phase transition phenomenon very similar to that described above for the PML problem. While this falls just short of proving that the Bethe approximation itself admits such a phase transition, it lends strong support in favour of this assertion.  

The bulk of our proof of the existence of phase transitions in the degree-$M$ optimization problems involves analyzing a certain probability distribution, denoted by $Q_{k,M}$, on the set of $k \times k$ non-negative integer matrices with all row and column sums equal to $M$. This probability distribution and our analysis of it via a discrete Gaussian approximation may be of independent interest.

The remainder of this paper is organized as follows. In Section~\ref{sec:pml}, we provide the definitions needed to describe the PML problem, after which we state and prove the corresponding phase transition phenomenon (Theorem~\ref{thm:thres}). The Bethe approximation is described in Section~\ref{sec:bethe}. This section also explains the notion of ``degree-$M$ lifted permanents'' defined by Vontobel \cite{Von13a}, which is used to define our degree-$M$ optimization problems. Section~\ref{sec:bpml_phase} contains a precise statement (Theorem~\ref{thm:thresB}) of the phase transition phenomenon in the degree-$M$ problems, the proof of which occupies much of the rest of the paper. In particular, Section~\ref{sec:QkM} collects together the properties of the probability distribution $Q_{k,M}$ that are used in the proof. The paper concludes in Section~\ref{sec:conclusion} with a discussion of the gap remaining in a rigorous proof of the phase transition phenomenon in the Bethe approximation. Some of the more technical proofs from Sections~\ref{sec:bethe}--\ref{sec:QkM} are presented in appendices.


\section{The PML Problem and a Phase Transition Phenomenon} \label{sec:pml}

We use $\Z_+$ and $\Z_{++}$, respectively, to denote the set of non-negative and positive integers. 
For $k \in \Z_{++}$, we use $[k]$ to denote the set $\{1,2,\ldots,k\}$. For any countable set $\cX$, we let $\Pi_{\cX}$ denote the set of all probability distributions on $\cX$: 
$$\Pi_{\cX} = \biggl\{\p = {(p(x))}_{x \in \cX}: \ p(x) \ge 0 \ \forall x \in \cX, \ \sum_{x \in \cX} p(x) = 1\biggr\}.$$
For any $k \in Z_{++}$, we let $U_k$ denote the uniform distribution on $[k]$.

\subsection{Patterns and PML} \label{sec:patterns}

Given a sequence $\x^n = x_1,\ldots,x_n$ over some alphabet, the \emph{pattern} of $\x^n$ is the sequence $\psibf = \psi_1,\psi_2,\ldots,\psi_n$ obtained by replacing each $x_j$ by the order of its first occurrence in $\x^n$ \cite{OSSVZ04}, \cite{Von12}. More precisely, for each symbol $x$ occurring in $\x^n$, let $\nu(x)$ denote the number of distinct symbols seen in the \emph{shortest} prefix of $\x^n$ that ends in the symbol $x$. Then, $\psi_j = \nu(x_j)$ for $j = 1,2,\ldots,n$. The pattern $\psibf(\x^n)$ is defined to have \emph{length} $n$ and \emph{size} $m$, where $m$ is the number of distinct symbols in $\x^n$.  For example, the word ``sleepless'' has pattern $123342311$, which is of length $9$ and size $4$. We will canonically represent a pattern $\psibf$ as $1^{\mu_1} 2^{\mu_2} \ldots m^{\mu_m}$, where $\mu_j$ is the \emph{multiplicity} of the symbol $j$, i.e., the number of times $j$ appears, in $\psibf$. Note that $\mu_1 +\cdots + \mu_m = n$. The pattern $\psibf$ in our example has canonical form $1^3 2^2 3^3 4$. 

Let $\psibf$ be a given pattern of length $n$, and let $\p = {(p(x))}_{x \in \cX}$ be a probability distribution over a discrete (possibly countably infinite) set $\cX$. The probability that $n$ i.i.d.\ samples drawn from the distribution $\p$ forms a sequence with pattern $\psibf$ is given by 
\begin{equation}
P(\psibf;\p) := \sum_{\x^n: \psibf(\x^n) = \psibf} \prod_{i=1}^n p(x_i).
\label{def:Ppsi}
\end{equation}
Clearly, all patterns $\psibf$ with the same canonical form $1^{\mu_1} 2^{\mu_2} \ldots m^{\mu_m}$ will have the same pattern probability $P(\psibf;\p)$. Indeed, if $\p = (p_1,\ldots,p_k) \in \Pi_{[k]}$ with $k \ge m$, then we can write
\begin{equation}
P(\psibf;\p) = \sum_{\sigma} \prod_{i=1}^m p_{\sigma(i)}^{\mu_i},
\label{eq:Ppsi}
\end{equation}
where the summation runs over all one-to-one maps $\sigma: [m] \to [k]$. 

The right-hand side of \eqref{eq:Ppsi} can be expressed in alternative form using the notion of a permanent of a matrix. The \emph{permanent} of a real $k \times k$ matrix $\Theta = (\theta_{i,j})$ is defined as 
$$
\perm(\Theta) = \sum_{\pi} \prod_{i=1}^k \theta_{i,\pi(i)},
$$
where the summation is over all permutations $\pi: [k] \to [k]$. With this, it can be verified that \eqref{eq:Ppsi} can be re-written as 
\begin{equation}
P(\psibf;\p) = \frac{1}{(k-m)!} \, \perm(\Theta(\psibf;\p)),
\label{eq:Ppsi_alt}
\end{equation}
where $\Theta(\psibf;\p)$ is the $k \times k$ matrix $(\theta_{i,j})$ with $\theta_{i,j} = p_i^{\mu_j}$; here, we set\footnote{For consistency, we define $0^0 = 1$. This is also in keeping with the convention used in Definition~\ref{def:permB} in Section~\ref{sec:bethe} that $0 \log 0 = 0$.} $\mu_j = 0$ for $m+1 \le j \le k$. The $\frac{1}{(k-m)!}$ term in \eqref{eq:Ppsi_alt} comes from the fact that each one-to-one map $\sigma: [m] \to [k]$ in the sum \eqref{eq:Ppsi} can be extended to a permutation $\pi:[k] \to [k]$ in exactly $(k-m)!$ different ways.


The \emph{PML probability} of a pattern $\psibf$ of size $m$ is defined as 
\begin{equation}
P^{\text{PML}}(\psibf) := \max_\p  P(\psibf;\p) 
\label{def:PPML}
\end{equation}
the maximum being taken over all discrete distributions $\p$ of support size at least $m$. 
Any distribution that attains the maximum above is called a \emph{PML distribution} of $\psibf$,
denoted by $\p^{\text{PML}}(\psibf)$. For the purposes of this paper, we will assume that
the maximum is indeed attained by some discrete distribution $\p$.\footnote{In general, to guarantee that 
the maximum is always attained, we must allow ``mixed'' distributions; see \cite{OSZ04}, \cite{OSVZ11}.} In this case, there is always a PML distribution with finite support \cite{OSVZ11}. Hence, we have 
\begin{equation}
P^{\text{PML}}(\psibf) = \max_{k \ge m} \max_{\p \in \Pi_{[k]}}  P(\psibf;\p)
\label{PPML_alt}
\end{equation}
It should be pointed out that, for any $k \ge m$, since $P(\psibf;\p)$ is a continuous function of $\p \in \Pi_{[k]}$, as is evident from \eqref{eq:Ppsi}, it does attain its maximum on the compact set $\Pi_{[k]}$. 

The problem of determining the PML distribution of a pattern seems to be computationally difficult in general \cite{OSVZ11}, \cite{OSSVZ04}--\cite{ADOP11}. Algorithms for approximating the PML distribution have been proposed by Orlitsky et al.\ \cite{OSSVZ04} and Vontobel \cite{Von12}. Vontobel's algorithm, in particular, uses the Bethe approximation, about which we will have much more to say in Section~\ref{sec:bethe}.

\subsection{Phase Transition in the PML Problem} \label{sec:pml_phase}

Consider now the potentially simpler decision problem of determining whether or not the PML distribution of a given pattern is a uniform distribution. A natural approach to this problem would be to find a test for whether, for any fixed $k \ge m$, the uniform distribution achieves the inner maximum in \eqref{PPML_alt}. In attempting this approach, we discovered a striking phase transition phenomenon in the PML problem. To describe this, we introduce some notation. For a pattern $\psibf$ of size $m$, and an integer $k \ge m$, let $\beta_k^{\psibf}: \Pi_{[k]} \to [0,1]$ be the function defined by the mapping $\p \mapsto P(\psibf;\p)$. The phase transition phenomenon is made precise in the following theorem. 

\begin{theorem}
For a pattern $\psibf$ of length $n$ with canonical form $1^{\mu_1} 2^{\mu_2} \ldots m^{\mu_m}$, $m \ge 2$, define
\begin{equation}
\Upsilon(\psibf) = \frac{n^2-n}{\sum_{i=1}^m \mu_i^2-n}. 
\label{eq:Theta}
\end{equation}
 Then, for all integers $k \ge m$, the following holds: when $k <\Upsilon(\psibf)$, the uniform distribution $U_k$ is a local maximum of the function $\beta_k^{\psibf}$, and when $k > \Upsilon(\psibf)$, $U_k$ is a local minimum.
\label{thm:thres}
\end{theorem}

We clarify a point concerning the statement of the theorem. Note that $\Upsilon(\psibf)$ is finite iff $\psibf \ne 123 \ldots n$. When $\psibf = 123 \ldots n$, we take $\Upsilon(\psibf)$ to be $\infty$. 

\begin{proof}[Proof of Theorem~\ref{thm:thres}]
The proof approach is based on that of Theorem~20 in \cite{Von13a}. Let $\p = U_k$, so that $\p$ is in the interior of the simplex $\Pi_{[k]}$. Pick an arbitrary direction $\xibf \in \R^k \setminus \{\0\}$, normalized so that ${\|\xibf\|}_2 = 1$, such that for all $t$ within a sufficiently small interval around $0$, the point $\p(t) = \p + t\,\xibf$ continues to lie within $\Pi_{[k]}$. Note that this implies that $\sum_{j=1}^k \xi_j = 0$. Consider the function $g(t) = P(\psibf;\p(t))$. We will show that, \emph{independent} of the choice of $\xibf$, we have $g'(0) = 0$, $g''(0) < 0$ if $k <\Upsilon(\psibf)$, and $g''(0) > 0$ if $k > \Upsilon(\psibf)$. This clearly suffices to prove the theorem.

Now, from \eqref{eq:Ppsi}, $g(t)$ is expressible as $\sum_{\sigma} g_{\sigma}(t)$, where 
$g_{\sigma}(t) = \prod_{i=1}^m (p_{\sigma(i)} + t \, \xi_{\sigma(i)})^{\mu_i}$. Differentiation, together with the fact that $p_j = \frac{1}{k}$ for all $j$, yields $g'_{\sigma}(0) = \frac{1}{k^{n-1}} \sum_{i=1}^m \mu_i \xi_{\sigma(i)}$. Hence, 
$$
g'(0)= \sum_{\sigma} g'_{\sigma}(0) = \frac{1}{k^{n-1}} \sum_{i=1}^m \mu_i \sum_{\sigma} \xi_{\sigma(i)}.
$$
For any fixed $i \in [m]$, the inner summation $\sum_{\sigma} \xi_{\sigma(i)}$ can be evaluated as follows. As $\sigma$ ranges over all one-to-one maps from $[m]$ to $[k]$, for each $j \in [k]$, $\sigma(i)$ takes the value $j$ exactly $\frac{(k-1)!}{(k-m)!}$ times. Hence, $\sum_{\sigma} \xi_{\sigma(i)} = \frac{(k-1)!}{(k-m)!} \, \sum_{j=1}^k \xi_j = 0$ by choice of $\xibf$.  Thus, $g'(0) = 0$.

Next, we compute $g''(0) = \sum_{\sigma} g''_{\sigma}(0)$. Straightforward computations yield 
$$
g''_\sigma(0) = \frac{1}{k^{n-2}} \left[ \left(\sum_{i=1}^m \mu_i \xi_{\sigma(i)}\right)^2 - \sum_{i=1}^m \mu_i \xi_{\sigma(i)}^2\right].
$$
Re-write the term within square brackets as 
$$
\sum_{i=1}^m \mu_i(\mu_i-1) \xi_{\sigma(i)}^2 + \sum_{(i,\ell): i \ne \ell} \mu_i \mu_{\ell} \xi_{\sigma(i)}\xi_{\sigma(\ell)}.
$$
Summing over all one-to-one maps $\sigma:[m] \to [k]$, we obtain
$$
\sum_{i=1}^m \mu_i(\mu_i-1) \sum_{\sigma} \xi_{\sigma(i)}^2+ \sum_{(i,\ell): i \ne \ell} \mu_i \mu_{\ell} \sum_{\sigma} \xi_{\sigma(i)}\xi_{\sigma(\ell)}.
$$
As above, $\sum_{\sigma} \xi_{\sigma(i)}^2 = \frac{(k-1)!}{(k-m)!} \sum_{j=1}^k \xi_j^2$ which, since ${\|\xibf\|}_2 = 1$, means that $\sum_{\sigma} \xi_{\sigma(i)}^2 = \frac{(k-1)!}{(k-m)!}$. Similarly, for $i \neq \ell$, 
$\sum_{\sigma} \xi_{\sigma(i)}\xi_{\sigma(\ell)} = \frac{(k-2)!}{(k-m)!}\sum_{(s,t) \in [k]^2: s \ne t} \xi_s \xi_t$. We also have 
$0 = \left(\sum_{j=1}^k \xi_j\right)^2$, from which we obtain
$\sum_{j=1}^k \xi_j^2 = - \sum_{(s,t) \in [k]^2: s \ne t} \xi_s \xi_t$.
Hence, $\sum_{\sigma} \xi_{\sigma(i)}\xi_{\sigma(\ell)} = - \frac{(k-2)!}{(k-m)!} $, again using the fact that ${\|\xibf\|}_2 = 1$. Putting it all together, we find that 
$$
g''(0) = C 
   \left[(k-1) \sum_{i=1}^m \mu_i(\mu_i-1) - \sum_{(i,\ell) \in [m]^2: i \ne \ell} \mu_i \mu_\ell \right],
$$
where $C = \frac{1}{k^{n-2}} \frac{(k-2)!}{(k-m)!}$ is a positive constant independent of $\xibf$.  Further simplification using the fact that $\sum_{i=1}^m \mu_i = n$ yields 
$$
g''(0) = C \left[ k \left(\sum_{i=1}^m \mu_i^2 - n\right) - (n^2-n)\right], 
$$
from which the desired result follows.
\end{proof}

Theorem~\ref{thm:thres} shows that the uniform distribution $U_k$ is either a local maximum or a local minimum of $\beta_k^{\psi}$ for all integers $k \ge m$, except perhaps at the threshold $\Upsilon(\psibf)$. Indeed, if $\Upsilon(\psibf)$ happens to be an integer, then for $k=\Upsilon(\psibf)$, it is possible that $U_k$ is not a local extremum, but only a saddle point. For example, for $\psibf = 1122$, we have $\Upsilon(\psibf) = 3$, and it may be verified that $U_3$ is not a local extremum for $\beta_3^{\psi}$.

As a simple corollary of the theorem, we see that a necessary condition for the PML distribution of a pattern $\psibf$ to be uniform is that $\Upsilon(\psibf) \ge m$. However, this condition is not sufficient in general. In \cite{FK13}, we derive a slightly stronger necessary condition using Theorem~\ref{thm:thres}, which is used as the basis for a heuristic algorithm that determines whether or not a given pattern has a uniform PML distribution. 

The intent of this paper, however, is to investigate whether the phase transition phenomenon reported in Theorem~\ref{thm:thres} extends to the Bethe approximation of the PML problem, which we describe in the next section. 

\section{The Bethe Approximation} \label{sec:bethe}
The Bethe approximation is a method whose origins lie in statistical physics \cite{Bethe35}, \cite{Peierls36}. In the interest of brevity, we describe this approximation only in the context of the PML problem. The motivation and justification behind the definitions in this section are discussed in detail in \cite{Von13a}. 

\subsection{The Bethe PML Problem} \label{sec:bethePML}

From \eqref{eq:Ppsi_alt}, we see that computing the pattern probability $P(\psibf;\p)$ is equivalent to computing the permanent of the matrix $\Th(\psibf;\p)$. It is well known that computing the permanent of a matrix is hard in general; formally, the problem is \#P-complete \cite{Valiant79}. Many approximation algorithms have been developed for this problem (e.g., \cite{JSV04}, \cite{HL08}), of which the ones based on the Bethe approximation \cite{CKV08}, \cite{HJ09}, \cite{Von13a} are relevant to us.  

Let $\cD_k$ denote the set of $k \times k$ doubly stochastic matrices. In the following definition, we use the convention that $0 \log 0 = 0$. 
\begin{definition}[\cite{Von13a}, Corollary~15]
The \emph{Bethe permanent} of a non-negative $k \times k$ matrix $\Theta = (\theta_{i,j})$, with $\theta_{i,j} \ge 0$ for all $i,j$, is defined as 
$$
\perm_B(\Theta) := \max_{\Gm \in \cD_k} \exp\left(- F_B(\Gm,\Theta) \right),
$$
where for $\Gm = (\gamma_{i,j}) \in \cD_k$, we have $F_B(\Gm) = U_B(\Gm,\Theta) - H_B(\Gm)$, with
\begin{align*}
  U_B(\Gm,\Theta) &= - \sum_{i,j} \gm_{i,j} \log(\theta_{i,j}), \\
  H_B(\Gm) &= - \sum_{i,j} \gm_{i,j} \log(\gm_{i,j}) + \sum_{i,j} (1-\gm_{i,j}) \log(1-\gm_{i,j}). 
\end{align*}
\label{def:permB}
\end{definition}

The function $F_B(\Gm,\Theta)$ in the above definition is called the \emph{Bethe free energy}. If the pair $(\Gamma,\Theta)$ is such that $\gamma_{i,j} > 0$ but $\theta_{i,j} = 0$ for some $(i,j)$, we define $F_B(\Gm,\Theta) = \infty$, and correspondingly, $\exp(-F_B(\Gamma,\Theta)) = 0$. With these definitions, $\exp(-F_B(\Gamma,\Theta))$ is a continuous function of $(\Gamma,\Theta)$, so that for any fixed $\Theta$, $\exp(-F_B(\Gamma,\Theta))$ attains a maximum on the compact set $\cD_k$. Hence, $\perm_B(\Theta)$ is well-defined. 

For positive matrices $\Theta$, we can write 
$$
\perm_B(\Theta) = \exp\left(  - \min_{\Gm \in \cD_k} F_B(\Gm,\Theta) \right).
$$
Vontobel \cite[Corollary~23]{Von13a} showed that for any positive matrix $\Theta$, $F_B(\Gamma,\Theta)$ is a convex function of $\Gm \in \cD_k$, so that $\min_{\Gm \in \cD_k} F_B(\Gm,\Theta)$ is a convex program. Vontobel further proved that the sum-product algorithm (belief propagation) can be used to find this minimum, and hence $\perm_B(\Theta)$, highly efficiently. Since the Bethe permanent is often a very good proxy for the actual permanent \cite{HJ09}, \cite{Von12}, \cite{Fer13}, having an efficient algorithm to compute it is particularly useful. 

For a pattern $\psibf$ of size $m$ and a probability distribution $\p \in \Pi_{[k]}$, $k \ge m$, we define, in analogy with \eqref{eq:Ppsi_alt}, the quantity
\begin{equation}
P_B(\psibf;\p) := \frac{1}{(k-m)!} \, \perm_B(\Theta(\psibf;\p)).
\label{def:BPPsi}
\end{equation}
We then have $0 \le P_B(\psibf;\p) \le P(\psibf;\p) \le 1$, the inequalities holding for the following reasons:
\begin{itemize}
\item the first inequality is simply a consequence of the non-negativity of the Bethe permanent;
\item the second inequality is because of the fact that $\perm(\Th) \ge \perm_B(\Th)$ for any non-negative matrix $\Th$, an inequality proved by Gurvits \cite{Gur11a}, \cite{Gur11b};
\item the last inequality is a consequence of the fact that $P(\psibf;\p)$ is a probability.
\end{itemize}
Thus, $P_B(\psibf;\p)$ can be viewed as a probability as well. 

With this, we define, in analogy with \eqref{def:PPML} and \eqref{PPML_alt}, the \emph{Bethe PML probability} of a pattern $\psibf$ to be 
\begin{equation}
P^{\BPML}(\psibf) :=  \sup_{k \ge m} \max_{\p \in \Pi_{[k]}} P_B(\psibf;\p).
\label{def:BPML}
\end{equation}
A couple of clarifications on this definition may be needed. One is that for any positive integer $k$, $\max_{\p \in \Pi_{[k]}} P_B(\psibf;\p)$ is well-defined. This is because $\chi(\Gamma,\p) := \exp\bigl\{-F_B(\Gamma,\Theta(\psibf;\p))\bigr\}$, as a function of $(\Gamma,\p)$, is continuous on the compact set $\cD_k \times \Pi_{[k]}$. Consequently, $\perm_B(\Theta(\psibf;\p)) = \max_{\Gamma \in \cD_k} \chi(\Gamma,\p)$ is a continuous function of $\p$. 
Hence, $P_B(\psibf;\p)$, being a continuous function of $\p$, must attain a maximum on the compact set $\Pi_{[k]}$.

A second clarification is that it is not known whether the supremum in \eqref{def:BPML} is always achieved at some finite $k$, although empirical evidence suggests that this may indeed be the case \cite{Fer13}. Empirically again, the \emph{Bethe PML distribution}, defined as any distribution $\p$ for which $P_B (\psibf;\p) = P^{\BPML}(\psibf)$, is a very good approximation of the PML distribution of a pattern $\psibf$. The ``Bethe PML problem'' of determining the Bethe PML distribution is also considerably easier to solve numerically  \cite{Von12}, \cite{Fer13}. 

The question we are interested in addressing is whether the Bethe PML problem exhibits a phase transition analogous to that described for the PML problem in Theorem~\ref{thm:thres}. To answer this, we must understand when the uniform distribution $U_k$ is a local maximum or a local minimum in $\Pi_{[k]}$ for the function $\p \mapsto \perm_B(\Theta(\psibf;\p))$. A direct approach analogous to that used in the proof of Theorem~\ref{thm:thres} seems difficult as we only have a description of $\perm_B$ as a solution to a convex optimization problem. Instead, we take an indirect approach via the degree-$M$ lifted permanents discussed next.

\subsection{Degree-$M$ Lifted Permanents}\label{sec:degreeM}
As an alternative to defining the Bethe permanent as a solution to an optimization problem, Vontobel gave a combinatorial characterization of this quantity, which we describe here. Let $\Theta = (\theta_{i,j})$ be a given $k \times k$ matrix with non-negative entries. For a positive integer $M$, let $\cP_M$ denote the set of all $M \times M$ permutation matrices. Further, let $\cP_M^{k \times k}$ be the set of all $kM \times kM$ matrices of the form 
\begin{equation}
\Lambda = 
\left( 
\begin{array}{cccc}
P^{(1,1)} & P^{(1,2)} & \cdots & P^{(1,k)} \\
P^{(2,1)} & P^{(2,2)} & \cdots & P^{(2,k)} \\
\vdots & \vdots & \ddots & \vdots \\
P^{(k,1)} & P^{(k,2)} & \cdots & P^{(k,k)}
\end{array}
\right)
\label{eq:Lambda}
\end{equation}
with $P^{(i,j)} \in \cP_M$ for all $i,j$. For a $\Lambda$ as above, define 
\begin{equation}
\Theta \odot \Lambda = 
\left( 
\begin{array}{cccc}
\th_{1,1} P^{(1,1)} & \th_{1,2} P^{(2,1)} & \cdots & \th_{1,k} P^{(1,k)} \\
\th_{2,1} P^{(2,1)} & \th_{2,2} P^{(2,2)} & \cdots & \th_{1,k} P^{(2,k)} \\
\vdots & \vdots & \ddots & \vdots \\
\th_{k,1} P^{(k,1)} & \th_{k,2} P^{(k,2)} & \cdots & \th_{k,k} P^{(k,k)}
\end{array}
\right)
\label{M-lift}
\end{equation}

\begin{definition}[\cite{Von13a}, Definition~38]
The \emph{degree-$M$ lifted permanent} of $\Th$ is defined to be 
$$
\perm_{B,M}(\Th) := {\bigl\langle\perm(\Th\odot\Lambda)\bigr\rangle}^{1/M},
$$
where the angular brackets represent the arithmetic average of $\perm(\Th \odot \Lambda)$ as $\Lambda$ ranges over the $(M!)^{k^2}$ matrices in $\cP_M^{k \times k}$. Equivalently, $\bigl\langle\perm(\Th\odot\Lambda)\bigr\rangle$ is the expected value of $\perm(\Th \odot \Lambda)$, the expectation being taken over $\Lambda$ chosen uniformly at random from $\cP_M^{k \times k}$. 
\label{def:permBM}
\end{definition}

Note that when $M=1$, $\perm_{B,M}(\Th)$ is equal to $\perm(\Th)$. At the other extreme, as $M \to \infty$, Vontobel \cite[Theorem~39]{Von13a} has shown the following identity:
\begin{equation}
\limsup_{M \to \infty} \perm_{B,M}(\Th) = \perm_B(\Th). 
\label{eq:Vontobel}
\end{equation}
Thus, degree-$M$ lifted permanents interpolate between $\perm(\Th)$ and $\perm_B(\Th)$. The advantage of using degree-$M$ lifted permanents as an indirect means of understanding the Bethe permanent is that they can be expressed in a form that is more amenable to analysis.

\begin{proposition} For any $k \times k$ matrix $\Th = (\th_{i,j})$ and any positive integer $M$, we have
$$
\left[\perm_{B,M} (\Th)\right]^M
 = (M!)^{2k-k^2} \sum_{(a_{i,j}) \in \cA_{k,M}} \prod_{(i,j) \in {[k]}^2} \theta_{i,j}^{a_{i,j}} \, \frac{(M-a_{i,j})!}{(a_{i,j})!},
$$
where $\cA_{k,M}$ denotes the set of all $k \times k$ non-negative integer matrices whose row and column sums are all equal to $M$. 
\label{prop:permBM}
\end{proposition}
Using multinomial coefficients, the identity above can be expressed in an alternative, more evocative form\footnote{It is also possible to recover this form from Lemma~29 in \cite{Von13b}.}:
\begin{equation}
\left[\perm_{B,M} (\Th)\right]^M
 = \sum_{(a_{i,j}) \in \cA_{k,M}} \left[\prod_{(i,j) \in {[k]}^2} \theta_{i,j}^{a_{i,j}} \right] \, \frac{\prod_{i=1}^k\binom{M}{a_{i,1},\ldots,a_{i,k}} \prod_{j=1}^k \binom{M}{a_{1,j},\ldots,a_{k,j}}}{\prod_{(i,j) \in [k]^2} \binom{M}{a_{i,j}}}.
\label{perm_multinom}
\end{equation}
Proposition~\ref{prop:permBM} is proved in Appendix~A.


\section{Phase Transition in the Bethe PML Problem}\label{sec:bpml_phase}

As mentioned at the end of Section~\ref{sec:bethePML}, we take an indirect approach, via degree-$M$ lifted permanents, to the question of the existence of a phase transition phenomenon in the Bethe PML problem. This approach is based on the intuition that the large-$M$ behaviour of these lifted permanents will, by virtue of \eqref{eq:Vontobel}, shed light on the behaviour of the Bethe permanent. With this program in mind, we define for a pattern $\psibf$ of size $m \ge 2$, and integers $k \ge m$ and $M \ge 1$, a function $\beta_{k,M}^{\psibf}: \Pi_{[k]} \to \R_+$ that maps $\p \in \Pi_{[k]}$ to $\perm_{B,M}(\Theta(\psibf;\p))$. Recall that $U_k$ denotes the uniform distribution on $[k]$. The aim of this section is to prove the following result. 

\begin{theorem}  
Let $\psibf$ be a pattern of length $n$ having canonical form $1^{\mu_1} 2^{\mu_2} \ldots m^{\mu_m}$, $m \ge 2$. There is a threshold $\Upsilon_B(\psibf)$ such that for all integers $k \ge m$, the following holds for all sufficiently large $M$:
\begin{itemize}
\item when $k < \Upsilon_B(\psibf)$, $U_k$ is a local maximum for $\beta_{k,M}^{\psibf}$; and
\item when $k > \Upsilon_B(\psibf)$, $U_k$ is a local minimum for $\beta_{k,M}^{\psibf}$.
\end{itemize}
When $m=2$, the threshold $\Upsilon_B(\psibf)$ may be chosen as follows:
$$
\Upsilon_B(\psibf) = 
\begin{cases}
\infty & \text{ if } \mu_1 = \mu_2 = 1 \\
2 + \delta & \text{ if } \mu_1 = \mu_2 > 1 \\
1 + \delta & \text{otherwise}
\end{cases}
$$
for any $\delta \in (0,1)$.

When $m \ge 3$, $\Upsilon_B(\psibf)$ may be chosen to closely mimic the threshold $\Upsilon(\psibf)$ of Theorem~\ref{thm:thres} in the following sense:
\begin{itemize}
\item if $\Upsilon(\psibf) = \infty$ (which happens iff $\psibf = 123 \ldots n$), then $\Upsilon_B(\psibf) = \infty$;
\item if $\Upsilon(\psibf) < \frac{\sqrt{n}+1}{\sqrt{n}-1}$, then we may choose $\Upsilon_B(\psibf) = \Upsilon(\psibf)$;
\item in all other cases, we may choose
$$
\Upsilon_B(\psibf) = \frac{\cU + n^2 - 2n + \sqrt{(n^2+2n-\cU)^2 - 4n^3}}{2(\cU-n)},
$$
where $\cU  = \sum_{i=1}^m \mu_i^2$, so that $\Upsilon(\psibf)-1 \le \Upsilon_B(\psibf) < \Upsilon(\psibf)$ holds.
\end{itemize}
\label{thm:thresB}
\end{theorem}

The theorem does not explicitly give a comparison between the thresholds $\Upsilon(\psibf)$ and $\Upsilon_B(\psibf)$ in the case when $m=2$. This has been done only so that a clean statement of the result could be given. Indeed, there is a close relationship between the two thresholds even in this case: it can be shown using Theorem~\ref{thm:thres} that when $m=2$,
$$
\Upsilon(\psibf) = 
\begin{cases}
\infty & \text{ if } \mu_1 = \mu_2 = 1 \\
2 + \frac{1}{\mu_1-1} & \text{ if } \mu_1 = \mu_2 > 1 
\end{cases}
$$
and $\Upsilon(\psibf)$ lies in the interval $(1,3)$ otherwise. 

In summary, Theorem~\ref{thm:thresB} strongly indicates that the Bethe PML problem exhibits a phase transition phenomenon very similar to that proved in Theorem~\ref{thm:thres} for the PML problem. Unfortunately, this does not quite prove that there is indeed such a phase transition in the Bethe PML problem. We make some remarks concerning this in Section~\ref{sec:conclusion}.

The rest of this section is devoted to a proof of Theorem~\ref{thm:thresB}. The proof proceeds along the same lines as that of Theorem~\ref{thm:thres}, except that the calculations are messier. To preserve the flow of this section, we have moved the proofs of some intermediate lemmas, which mainly involve tedious calculations, to the appendices. Also, Proposition~\ref{prop:VarkM} below, which is also an intermediate step in the proof of Theorem~\ref{thm:thresB}, but which could be considered an interesting result in its own right, is proved in Section~\ref{sec:QkM}.

Let $\p = U_k$, and pick an arbitrary direction $\xibf \in \R^k \setminus \{\0\}$, normalized so that ${\|\xibf\|}_2 = 1$, such that for all $t$ within a sufficiently small interval around $0$, the point $\p(t) = \p + t\,\xibf$ continues to lie within $\Pi_{[k]}$. Note that this implies that $\sum_{j=1}^k \xi_j = 0$. 

Given a pattern $\psibf$ with multiplicities $(\mu_1,\ldots,\mu_m)$, $m \ge 2$, and integers $k \ge m$ and $M \ge 1$, define the function $G_{k,M}(t) = \beta_{k,M}^{\psibf}(\p(t))$.\footnote{Here, and for the remainder of this section, we will suppress the dependence on $\psibf$ in our notation; thus, we write $G_{k,M}(t)$ instead of $G_{k,M}^{\psibf}(t)$, $\Upsilon_B$ instead of $\Upsilon_B(\psibf)$, and so on.} As in the proof of Theorem~\ref{thm:thres}, the idea is to show that $G_{k,M}'(0) = 0$, and that for a suitable choice of $\Upsilon_B$ independent of $\xibf$, the sign of $G_{k,M}''(0)$ depends, for all sufficiently large $M$, only on whether $k < \Upsilon_B$ or $k > \Upsilon_B$. We will in fact prove the statement about the second derivative in the following equivalent form: there exists a threshold $\Upsilon_B$ independent of $\xibf$ such that
\begin{equation}
\lim_{M \to \infty}G_{k,M}''(0) < 0 \text{ for all } k < \Upsilon_B \ \ \ \text{ and } \ \ \ 
\lim_{M \to \infty}G_{k,M}''(0) > 0 \text{ for all } k > \Upsilon_B
\label{eq:limG''}
\end{equation}

It is straightforward to show that $G_{k,M}'(0) = 0$; a proof of this will be given in Appendix~B as part of the proof of Lemma~\ref{lem:G} below. To express $G_{kM}''(0)$, we consider once again\footnote{See Proposition~\ref{prop:permBM}.} the set, $\cA_{k,M}$, of $k \times k$ non-negative integer matrices all of whose row and column sums are equal to $M$. For $A = (a_{i,j}) \in \cA_{k,M}$, define 
\begin{equation}
w(A) = \prod_{(i,j) \in {[k]}^2} \frac{(M-a_{i,j})!}{(a_{i,j})!}
\label{def:w}
\end{equation}
and let 
\begin{equation}
Z_{k,M} := \sum_{A \in \cA_{k,M}} w(A).
\label{def:ZkM}
\end{equation}
Then, $Q_{k,M}(A) := \frac{1}{Z_{k,M}} w(A)$ defines a probability distribution on $\cA_{k,M}$. We will study this probability distribution in more detail in the next section. For now, we use it to give an expression for $G_{k,M}''(0)$.


\begin{lemma}  We have $G_{k,M}'(0) = 0$ and 
$$
G_{k,M}''(0) = {[(M!)^{2k-k^2}]}^{\frac1M} (Z_{k,M})^{\frac1M} k^{1-n} \left[\frac{k^2}{(k-1)^2} \frac{\Var_{k,M}(a_{1,1})}{M} \biggl(k \sum_{i=1}^m \mu_i^2 - n^2\biggr) - n \right],
$$
where $\Var_{k,M}(a_{1,1})$ denotes the variance of the entry $a_{1,1}$ in a random matrix $A \in \cA_{k,M}$ chosen according to the distribution $Q_{k,M}$.
\label{lem:G}
\end{lemma}
The proof of the lemma is deferred to Appendix~B.  The quantities $Z_{k,M}$ and $\Var_{k,M}(a_{1,1})$ in the above expression for $G_{k,M}''(0)$ can be determined explicitly for $k=2$, and asymptotically as $M \to \infty$ for $k\ge 3$. 

\begin{proposition}
\begin{itemize}
\item[(a)] $Z_{2,M} = M+1$ and $\Var_{2,M}(a_{1,1}) = \frac{1}{12}M(M+2)$.
\item[(b)] For $k \ge 3$, we have
$$
\lim_{M \to \infty}  {[(M!)^{2k-k^2}]}^{\frac1M} (Z_{k,M})^{\frac1M} = \frac{(k-1)^{k(k-1)}}{k^{k(k-2)}}
$$
and
$$
\lim_{M \to \infty} \frac{1}{M} \Var_{k,M}(a_{1,1}) = \frac{(k-1)^3}{k^3(k-2)}.
$$
\end{itemize}
\label{prop:VarkM}
\end{proposition}

The proof of the proposition will be given in the next section. As a direct consequence of Lemma~\ref{lem:G} and Proposition~\ref{prop:VarkM}, we have the following result.

\setcounter{corollary}{0}
\begin{corollary}
\begin{itemize}
\item[(a)] When $k=m=2$, 
$$
\lim_{M \to \infty} G_{2,M}''(0) = 
\begin{cases}
-n \, 2^{1-n} & \text{ if } \mu_1 = \mu_2 \\
+\infty & \text{ if } \mu_1 \ne \mu_2.
\end{cases}
$$
\item[(b)] When $k \ge 3$ (and $k \ge m$),
\begin{equation}
\lim_{M \to \infty} G_{k,M}''(0) = \frac{(k-1)^{k(k-1)}}{k^{k(k-2)}} \, k^{1-n} \left[\frac{k-1}{k(k-2)} \biggl(k \sum_{i=1}^m \mu_i^2 - n^2\biggr) - n \right].
\label{eq:G''}
\end{equation}
\end{itemize}
\label{cor:G''}
\end{corollary}
\begin{proof}
Only part~(a) requires a note of explanation. When $k=m=2$, the term $k\sum_{i=1}^m \mu_i^2 - n^2$ in the expression for $G_{k,M}''(0)$ in Lemma~\ref{lem:G} reduces to $2(\mu_1^2 + \mu_2^2) -  (\mu_1+\mu_2)^2$, which equals $(\mu_1-\mu_2)^2$. 
\end{proof}

For our purposes, it is only the sign of $\lim_{M \to \infty} G_{k,M}''(0)$ that matters, so we will make much use of the weaker corollary below.

\begin{corollary}
\begin{itemize} 
\item[(a)] When $k = m= 2$, we have $\lim_{M\to\infty} G_{2,M}''(0) < 0$ if $\mu_1 = \mu_2$, and $\lim_{M\to\infty} G_{2,M}''(0) > 0$ if $\mu_1 \ne \mu_2$.
\item[(b)] When $k \ge 3$, we have $\lim_{M\to\infty}G_{k,M}''(0) \lessgtr 0$ if 
$$
k^2(\cU - n) - k(\cU+n^2-2n) + n^2 \lessgtr 0
$$
\end{itemize}
\label{cor:signG''}
\end{corollary}
\begin{proof}
It suffices to point out that the condition in part (b) above is equivalent to the term within square brackets in \eqref{eq:G''} being negative or positive. 
\end{proof}

Thus, when $k \ge 3$, the sign of $\lim_{M\to\infty}G_{k,M}''(0)$ depends only on where $k$ lies in relation to the roots of the quadratic polynomial $x^2(\cU-n) - x(\cU+n^2-2n) + n^2$. Note that $\cU = \sum_{i=1}^m \mu_i^2 \ge \sum_{i=1}^m \mu_i = n$, with equality iff $\mu_i = 1$ for all $i \in [m]$, i.e., $\psibf = 123 \ldots n$.
The lemma below summarizes the behaviour of the roots of the quadratic equation.

\begin{lemma} Write $q(x) = x^2(\cU-n) - x(\cU+n^2-2n) + n^2$, which has discriminant 
$D = (n^2+2n-\cU)^2 - 4n^3$. Recall that $\Upsilon = \frac{n^2-n}{\cU-n}$.
\begin{itemize}
\item[(1)] If $\cU = n$ (which happens iff $\psibf = 123 \ldots n$), then $q(x) = -x(n^2-n) + n^2$ has $\frac{n}{n-1}$ as its only root. Since $\frac{n}{n-1} \le 2$, we have $q(k) < 0$ for all $k \ge 3$.

\item[(2)] If $\cU > n$, then we have exactly one of the following two cases:
\begin{itemize}
\item[(a)] The discriminant $D$ is strictly negative, which happens iff $\Upsilon < \frac{\sqrt{n}+1}{\sqrt{n}-1}$, so that $q(x)$ has no roots. In this case, $q(k) > 0$ for all $k$.
\item[(b)] The quadratic has two real roots $\rho_1 \le \rho_2$ given by
$$
\rho_1 = \frac{\cU + n^2 - 2n - \sqrt{D}}{2(\cU-n)} \ \text{ and } \ 
\rho_2 = \frac{\cU + n^2 - 2n + \sqrt{D}}{2(\cU-n)}.
$$
In this case, we have $1 < \rho_1 \le 2$, so that $q(k) < 0$ for $3 \le k < \rho_2$, and $q(k) > 0$ for all $k > \rho_2$. Furthermore, $\Upsilon - 1 \le \rho_2 < \Upsilon$ holds.
\end{itemize}
\end{itemize}
\label{lem:roots}
\end{lemma}

The proof of the lemma is given in Appendix~C. We now have the tools required to complete the proof of Theorem~\ref{thm:thresB}.

\begin{proof}[Proof of Theorem~\ref{thm:thresB}]
Recall that the goal is to show that there is a threshold $\Upsilon_B$ independent of the direction vector $\xibf$ such that \eqref{eq:limG''} holds. Corollary~\ref{cor:G''} shows that $\lim_{M \to \infty} G_{k,M}''(0)$ is independent of $\xibf$. With this, the $m \ge 3$ case of Theorem~\ref{thm:thresB} follows directly from Corollary~\ref{cor:signG''}(b) and Lemma~\ref{lem:roots}.

The $m=2$ case requires a few additional details to be checked. When $\mu_1 = \mu_2 = 1$, Corollary~\ref{cor:signG''} and Lemma~\ref{lem:roots}(1) show that $\lim_{M \to \infty} G_{k,M}''(0) < 0$ for all $k \ge 2$, and hence, we can take $\Upsilon_B = \infty$.

When $\mu_1 = \mu_2 > 1$, we need to argue that \eqref{eq:limG''} holds for any $\Upsilon_B \in (2,3)$. Corollary~\ref{cor:signG''}(a) shows that $\lim_{M \to \infty} G_{k,M}''(0) < 0$ for $k=2$. Therefore, it remains to show that $\lim_{M \to \infty} G_{k,M}''(0) > 0$ for all $k \ge 3$. Similarly, when $\mu_1 \ne \mu_2$, we need to show that $\lim_{M \to \infty} G_{k,M}''(0) > 0$ for all $k \ge 2$, so that we can take $\Upsilon_B \in (1,2)$ in this case. Again, Corollary~\ref{cor:signG''}(a) takes care of the $k=2$ case, so we are left with $k \ge 3$. In summary, we must show that if $(\mu_1,\mu_2) \ne (1,1)$, then $\lim_{M \to \infty} G_{k,M}''(0) > 0$ for all $k \ge 3$. For this, we will appeal to Lemma~\ref{lem:roots}(2).

If the discriminant $D$ is negative, then we are done by Lemma~\ref{lem:roots}(2)(a). So, we may assume that $D \ge 0$, in which case the situation of Lemma~\ref{lem:roots}(2)(b) applies. It suffices to show that when $(\mu_1,\mu_2) = (a,b) \ne (1,1)$, then $\rho_2 \le 2$. Note that $D = (2(a+b) + 2ab)^2 - 4(a+b)^3 = 4[(a+b+ab)^2 - (a+b)^3]$. Using the fact that $(a+b)^3  \ge  4ab(a+b)$, we obtain $D \le 4[(a+b+ab)^2 - 4ab(a+b)] = 4(a+b-ab)^2$. Thus, $\sqrt{D} \le 2|a+b-ab|$, and hence,
\begin{equation}
\rho_2 = 1 + \frac{n^2 - \cU + \sqrt{D}}{2(U-n)} \le 1+\frac{ab + |a+b-ab|}{a^2+b^2-(a+b)}.
\label{rho2_ineq}
\end{equation}

Now, for positive integers $a,b$, we have $ab < a+b$ only if $b = 1$. If $(a,b) = (2,1)$, then observe that $D = -8 < 0$, which cannot happen. So we can have $ab < a+b$ only if $b = 1$ and $a \ge 3$. In this case, \eqref{rho2_ineq} becomes $\rho_2 \le 1 + \frac{a+1}{a^2-a} < 2$, the last inequality holding for any $a > 1+\sqrt{2}$. 

Finally, if $ab \ge a+b$, then \eqref{rho2_ineq} reduces to 
$$
\rho_2 \le 1+\frac{2ab - (a+b)}{a^2+b^2-(a+b)},
$$
which is at most $2$. 
We have thus shown that $\rho_2 \le 2$ whenever $(\mu_1,\mu_2) \ne (1,1)$, which completes the proof of the $m=2$ case of the theorem.
\end{proof}


\section{The Probability Distribution $Q_{k,M}$} \label{sec:QkM}

The main aim of this section is to prove Proposition~\ref{prop:VarkM}. For convenience, we recall the relevant definitions first. We use $\cA_{k,M}$ to denote the set of $k \times k$ non-negative integer matrices whose row and column sums are all equal to $M$. The probability distribution $Q_{k,M}$ on $\cA_{k,M}$ is defined by setting, for $A \in \cA_{k,M}$, $Q_{k,M}(A) = \frac{1}{Z_{k,M}} w(A)$, where $w(A)$ and $Z_{k,M}$ are given by \eqref{def:w} and \eqref{def:ZkM}. We are especially interested in determining $\Var_{k,M}(a_{1,1})$, the variance of the entry $a_{1,1}$ in a random matrix $A = (a_{i,j}) \in \cA_{k,M}$ chosen according to the distribution $Q_{k,M}$. 

To compute the variance of $a_{1,1}$, we need to know the mean $\E[a_{1,1}]$, where $\E[\cdot]$ denotes expectation with respect to the probability distribution $Q_{k,M}$. As the following lemma shows, this is quite straightforward.

\begin{lemma}
For any entry $a_{i,j}$ of a random matrix $A \in \cA_{k,M}$ chosen according to $Q_{k,M}$, we have
$$
\E[a_{i,j}] = \frac{1}{k} M.
$$
\label{lem:Eaij}
\end{lemma}
\begin{proof}
Note that $w(A)$, and consequently, the probability function $Q_{k,M}(A)$, is invariant to row and column permutations of $A$. Therefore, the expected value $\E[a_{i,j}]$ must be a constant independent of $(i,j)$. This constant can be explicitly evaluated as follows:
\begin{equation*}
\E[a_{i,j}] = \frac{1}{k} \sum_{\ell=1}^k \E[a_{i,\ell}] = \frac1k \E\biggl[\sum_{\ell=1}^k a_{i,\ell}\biggr] = \frac1k \, M,
\end{equation*}
the last equality using the fact that $\sum_{\ell=1}^k a_{i,\ell} = M$.
\end{proof}

The invariance of $Q_{k,M}(A)$ to row and column permutations of $A$ also implies that the variance of $a_{i,j}$ is independent of $(i,j)$. In spite of this, the explicit computation of $\Var_{k,M}(a_{1,1})$ seems difficult in general, with the notable exception of the case when $k=2$.

Observe that $\cA_{2,M}$ consists precisely of the matrices
$$
\left[
\begin{array}{cc}
a & M-a \\
M-a & a
\end{array}
\right]
$$
with $a \in \{0,1,\ldots,M\}$. For any such matrix $A$, we have $w(A) = 1$, and hence, $Z_{2,M} = |\cA_{2,M}| = M+1$. In particular, $Q_{2,M}$ is just the uniform distribution on $\cA_{2,M}$. It follows that $\Var_{2,M}(a_{1,1})$ is the variance of a random variable uniformly distributed on $\{0,1,\ldots,M\}$. An easy calculation shows this to be $\frac{1}{12} M(M+2)$, thus completing the proof of part~(a) of Proposition~\ref{prop:VarkM}.

\medskip

Henceforth, we consider the case when $k \ge 3$. Our interest is in the regime where $k$ is fixed and $M$ goes to $\infty$.

\subsection{Estimating $Z_{k,M}$} \label{sec:ZkM}

To estimate the normalization constant $Z_{k,M}$, we make use of the fact that 
$$
w^*_{k,M} \le Z_{k,M} \le w^*_{k,M} |\cA_{k,M}|,
$$
where $w^*_{k,M} = \max_{A \in \cA_{k,M}} w(A)$. As we will see below, it is possible to explicitly determine $w^*_{k,M}$. In the regime of interest to us, $w^*_{k,M}$ grows super-exponentially in $M$. Since $|\cA_{k,M}| \le (M+1)^{k^2}$, a quantity that is polynomial in $M$, we see that the asymptotics of $Z_{k,M}$ are governed by the asymptotics of $w^*_{k,M}$. 

Write $M = qk + r$, where $q = \lfloor M/k \rfloor$ and $r = M-kq$. Note that we have $0 \le r < k$. Let $\u = (u_1,\ldots,u_k) \in \Z_+^k$ be defined by $u_i = q+1$ for $1 \le i \le r$, and $u_i = q$ for $r+1 \le i \le k$. Clearly, $\sum_i u_i = M$. Let $U$ be the circulant matrix having $\u$ as its first row. The fact that $U$ is a circulant matrix assures us that it is in $\cA_{k,M}$.

\begin{proposition}
The matrix $U$ maximizes $w(A)$ among all $A \in \cA_{k,M}$, and hence, 
$$w^*_{k,M} = w(U) = \left[\frac{\bigl(M-(q+1)\bigr)!}{(q+1)!}\right]^{kr} \left[\frac{(M-q)!}{q!}\right]^{k(k-r)}.$$
\label{prop:U}
\end{proposition}

Before giving a proof of the proposition, we briefly discuss its implications. In the regime of fixed $k$ and $M \to \infty$, 
routine algebraic manipulations using Stirling's approximation yield
\begin{equation}
\lim_{M \to \infty} {[(M!)^{2k-k^2}]}^{\frac1M} (w^*_{k,M})^{\frac1M} = \frac{(k-1)^{k(k-1)}}{k^{k(k-2)}}.
\label{eq:lim_wkM}
\end{equation}
Thus, $w^*_{k,M}$ grows super-exponentially in $M$, so that as observed above, its asymptotic behaviour governs that of $Z_{k,M}$. Consequently, in the left-hand side of \eqref{eq:lim_wkM}, we may replace $w^*_{k,M}$ with $Z_{k,M}$, thereby obtaining the first equality stated in Proposition~\ref{prop:VarkM}(b).

\medskip

Our proof of Proposition~\ref{prop:U} is based upon the theory of majorization \cite{MO_book}. For any vector $\x = (x_1,\ldots,x_k) \in \R^k$, let $\x_{\downarrow} = (x_{[1]},\ldots,x_{[k]})$ denote the permutation of the components of $\x$ such that $x_{[1]} \ge \ldots \ge x_{[k]}$. A vector $\x \in \R^k$ is said to be \emph{majorized} by a vector $\y \in \R^k$, denoted by $\x \preceq \y$, if for $\ell = 1,\ldots,k-1$, we have
$$
\sum_{i=1}^{\ell} x_{[i]} \le \sum_{i=1}^{\ell} y_{[i]},
$$
and $\sum_{i=1}^k x_{[i]} = \sum_{i=1}^k y_{[i]}$ (or equivalently, $\sum_{i=1}^k x_i = \sum_{i=1}^k y_i$).

The next lemma plays an important role in our proof of Proposition~\ref{prop:U}. Let us define $S_{k,M}$ to be the set of vectors $\x = (x_1,\ldots,x_k) \in \Z_+^k$ such that $\sum_{i=1}^k x_i = M$. Also, recall from above our definition of the vector $\u = (u_1,\ldots,u_k) \in S_{k,M}$ with $u_i = q+1$ for $1 \le i \le r$, and $u_i = q$ for $r+1 \le i \le k$, where $q = \lfloor M/k \rfloor$ and $r = M-kq$.

\begin{lemma}
The vector $\u$ is majorized by every $\x \in S_{k,M}$.
\label{lem:u}
\end{lemma}
\begin{proof}
Consider an arbitrary $\x \in S_{k,M}$. By definition, $\sum_{i=1}^k u_i = \sum_{i=1}^k x_i = M$, so we must show that $\sum_{i=1}^{\ell} u_{[i]} \le \sum_{i=1}^{\ell} x_{[i]}$ for $\ell = 1,\ldots,k-1$. Suppose that $\sum_{i=1}^{\ell} u_{[i]} > \sum_{i=1}^{\ell} x_{[i]}$ for some $\ell \in \{1,\ldots,k-1\}$. Note that $\sum_{i=1}^{\ell} u_{[i]} = \ell q + \min(\ell,r)$. Thus, 
$$
\ell x_{[\ell+1]} \le \sum_{i=1}^{\ell} x_{[i]} < \sum_{i=1}^{\ell} u_{[i]} = \ell q + \min(\ell,r),
$$
from which we infer that $x_{[\ell+1]} < q + \min(1,r/\ell)$. Since $x_{[\ell+1]}$ must be an integer, we deduce that $x_{[\ell+1]} \le q$. 

On the other hand, we also have
\begin{align*}
(k-\ell) x_{[\ell+1]} \ge \sum_{i=\ell+1}^k x_{[i]} & = M-\sum_{i=1}^\ell x_{[i]} \\
 & > M - \sum_{i=1}^{\ell} u_{[i]} \\
 & = qk+r- \bigl(\ell q + \min(\ell,r)\bigr) \\
 & \ge q(k-\ell)
\end{align*}
Hence, $x_{[\ell+1]} > q$, which contradicts the inequality $x_{[\ell+1]} \le q$ deduced previously. Therefore, our assumption that $\sum_{i=1}^{\ell} u_{[i]} > \sum_{i=1}^{\ell} x_{[i]}$ cannot hold. 
\end{proof}

A function $\phi: \Z_+^k \to \R$ is said to be \emph{Schur-concave} if for all $\x,\y \in \Z_+^k$,
$$
\x \preceq \y \Longrightarrow \phi(\x) \ge \phi(\y).
$$
Note that if $\x$ and $\y$ are permutations of each other, then $\x_{\downarrow} = \y_{\downarrow}$, so that both $\x \preceq \y$ and $\y \preceq \x$ hold. Hence, a Schur-concave function $\phi$ must be \emph{symmetric}, which means that $\phi(\x) = \phi(\y)$ whenever $\x$ and $\y$ are permutations of each other. 

A characterization of Schur-concave functions is given in \cite[Chapter~3, A.2.b]{MO_book}, which we adapt to our context in the proposition below.

\begin{proposition}
A function $\phi: \Z_+^k \to \R$ is Schur-concave iff it is symmetric, and for each choice of non-negative integers $s,x_3,\ldots,x_k$, the function $\phi(x,s-x,x_3,\ldots,x_k)$ is monotonically decreasing in $x$ for $x \ge s/2$.
\label{prop:Schur}
\end{proposition}

Now, define the function $\phi: \Z_+^k \to \R$ as follows:
\begin{equation}
\phi(x_1,\ldots,x_k) = \prod_{j=1}^k \frac{\left(\sum_{\ell \neq j} x_\ell\right)!}{x_j!}.
\label{eq:phi}
\end{equation}
An application of Proposition~\ref{prop:Schur} shows that this function $\phi$ is Schur-concave. Indeed, $\phi$ is obviously symmetric. We claim that, for any choice of non-negative integers $s,x_3,\ldots,x_k$, the function $\varphi(x) := \phi(x,s-x,x_3,\ldots,x_k)$ is monotonically decreasing in $x$ for integers $x \ge s/2$. To see this, first verify that 
$$
\frac{\varphi(x+1)}{\varphi(x)} = \frac{(x+1 + x_3 + \ldots + x_k)(s-x)}{(x+1)(s-x + x_3+\ldots+x_k)}. 
$$
Now, we must show that for $x \ge s/2$, this ratio is at most $1$, or equivalently, $(x+1 + x_3 + \ldots + x_k)(s-x) - (x+1)(s-x + x_3+\ldots+x_k) \le 0$. Upon some re-arrangement and cancellation of like terms, the left-hand side becomes $(x_3 + \ldots + x_k)(s-x - (x+1))$, which is at most $0$ when $x \ge s/2$. 

\begin{proof}[Proof of Proposition~\ref{prop:U}]
For any $A \in \cA_{k,M}$, let $\a_1, \ldots,\a_k$ denote the rows of $A$. Note that $w(A) = \prod_{i=1}^k \phi(\a_i)$,
where $\phi$ is as defined in \eqref{eq:phi}. 

Since all row sums of $A$ are equal to $M$, the row vectors $\a_1,\ldots,\a_k$ are all in $S_{k,M}$. By Lemma~\ref{lem:u} and Schur-concavity of $\phi$, we have 
$$
w(A) = \prod_{i=1}^k \phi(\a_i) \le \prod_{i=1}^k \phi(\u) = w(U),
$$
which proves the proposition.
\end{proof}

We remark that with a little bit of extra work, it can be shown that when $k$ divides $M$ (and $k \ge 3$), the matrix $U$, which in this case has all entries equal to $M/k$, uniquely maximizes $w(A)$ among $A \in \cA_{k,M}$. When $k$ does not divide $M$, any matrix obtained by permuting the rows and columns of $U$ is also a maximizer.

\subsection{Estimating $\Var_{k,M}(a_{1,1})$}\label{sec:varkM}

Recalling that the variance of $a_{i,j}$ is independent of $(i,j)$, we have
$$
\Var_{k,M}(a_{1,1}) = \frac{1}{k^2} \sum_{i,j} \Var_{k,M}(a_{i,j}) = \frac{1}{k^2}\E\biggl[\sum_{i,j} (a_{i,j}-M/k)^2\biggr].
$$
We will recast this in terms of the matrix $U$ defined previously. Recall that this matrix has all its entries $u_{i.j}$ equal to either $\lfloor{M/k}\rfloor$ or $\lceil{M/k}\rceil$. 

For any matrix $T = (t_{i,j})$, we use $\|T\|^2$ to denote the sum $\sum_{i,j}t_{i,j}^2$. Observe that 
\begin{align*}
\E[\|A-U\|^2] & = \E\biggl[\sum_{i,j} (a_{i,j}-u_{i,j})^2\biggr]  \\
& = \E\biggl[\sum_{i,j} \bigl((a_{i,j}-M/k) - (u_{i,j}-M/k)\bigr)^2\biggr] \\
& = \E\biggl[\sum_{i,j} (a_{i,j}-M/k)^2\biggr] + \E\biggl[\sum_{i,j} (u_{i,j}-M/k)^2\biggr]. \\
\end{align*}
Hence, 
\begin{equation}
 \biggl| \Var_{k,M}(a_{1,1}) - \frac{1}{k^2} \E[\|A-U\|^2] \biggr| = \frac{1}{k^2}\E\biggl[\sum_{i,j} (u_{i,j}-M/k)^2\biggr] < 1.
\label{eq:E[A-U]}
\end{equation}
Thus, for our purposes, it suffices to estimate $\E[\|A-U\|^2]$. 

Our strategy to estimate $\E[\|A-U\|^2]$ is to approximate the distribution $Q_{k,M}$ by a suitably chosen discrete Gaussian distribution. The implementation of this strategy rests upon the following lemma.

\begin{lemma} Fix $k \ge 3$. Let $\rho = M/k$, and for $A \in \cA_{k,M}$, let $T = (t_{i,j}) = A-U$. If $\rho \ge 4$ and $\max_{i,j} |t_{i,j}| \le \frac19 \rho$, then 
$$
\left| \log \frac{w(A)}{w(U)} - \log \frac{\tw(A)}{\tw(U)} \right| \le \frac{4}{\rho^2} \sum_{i,j} (|t_{i,j}|+1)^3 + \frac{3}{2\rho} \sum_{i,j} |t_{i,j}|,
$$
where $\tw(A) := \exp\left\{ -\frac12 \cdot \frac{k-2}{k-1} \cdot \frac{1}{\rho} \sum_{i,j} t_{i,j}^2\right\}$.
\label{lem:gauss}
\end{lemma}

The proof of the lemma is given in Appendix~D.

Taking cue from Lemma~\ref{lem:gauss}, we define a discrete Gaussian measure on the set, $\tcA_{k,M}$, of all $k \times k$ integer (not necessarily non-negative) matrices $A = (a_{i,j})$ with row and column sums equal to $M$.  For any such matrix $A$, define $\tw(A) :=  \exp\left\{ -\frac{1}{2\sigma_{k,M}^2} \sum_{i,j} t_{i,j}^2\right\}$, where $(t_{i,j}) = A - U$ and $\sigma_{k,M}^2:= \frac{k-1}{k-2} \bigl(\frac{M}{k}\bigr)$. We then define a probability measure $\tQ_{k,M}$ on $\tcA_{k,M}$ as follows: for $A \in \tcA_{k,M}$,
$$
\tQ_{k,M}(A) := \frac{1}{\tZ_{k,M}} \tw(A),
$$
where $\tZ_{k,M} = \sum_{A \in \tcA_{k,M}} \tw(A)$. 

Let $\tE[\cdot]$ denote expectation with respect to the measure $\tQ_{k,M}$. To be clear, when we write $\tE[\|A-U\|^2]$, we mean $\sum_{A \in \tcA_{k,M}} \|A-U\|^2 \, \tQ_{k,M}(A)$, while $\E[\|A-U\|^2]$ refers to $\sum_{A \in \cA_{k,M}} \|A-U\|^2 \, Q_{k,M}(A)$. The following lemma, proved in Appendix~E, shows that $\E[\|A-U\|^2]$ is well-approximated by $\tE[\|A-U\|^2]$ as $M \to \infty$.

\begin{lemma}
$\E[\|A-U\|^2] = \tE[\|A-U\|^2]  + o(M)$.
\label{lem:E_and_tE}
\end{lemma}

The usefulness of this lemma stems from the fact that $\tE[\|A-U\|^2]$ can be estimated very accurately. To do this, we express $\tE[\|A-U\|^2]$ in a different form. For any $A \in \tcA_{k,M}$, note that $A-U$ is an integer matrix all of whose row and column sums are equal to $0$, i.e., $A-U \in \tcA_{k,0}$. For simplicity, let $\cT_k := \tcA_{k,0}$. We then have $\tE[\|A-U\|^2] = \frac{1}{\tZ} \sum_{T\in\cT_k} \|T\|^2 \exp\bigl\{-\frac{1}{2\sigma_{k,M}^2}\|T\|^2\bigr\}$, where $\tZ := \sum_{T\in\cT_k} \exp\bigl\{-\frac{1}{2\sigma_{k,M}^2}\|T\|^2\bigr\}$. 

For any $T = (t_{i,j}) \in \cT_k$, each of the entries in the $k$th row and column of $T$ is linearly dependent on the entries $t_{i,j}$, $1 \le i,j \le k-1$. To be precise,
\begin{align*}
t_{i,k} & = -\sum_{j=1}^{k-1} t_{i,j},  \ \ i = 1,\ldots,k-1 ; \\
t_{k,j} & = -\sum_{i=1}^{k-1} t_{i,j},  \ \ j = 1,\ldots,k-1 ;\\
t_{k,k} & = - \sum_{i=1}^{k-1} t_{i,k} = \sum_{i=1}^{k-1}\sum_{j=1}^{k-1} t_{i,j}.
\end{align*}
It follows that $\|T\|^2$ is a positive definite quadratic form in the $(k-1)^2$ variables $t_{i,j}$, $1 \le i,j \le k-1$. Hence, we can write\footnote{To avoid notational confusion, we write $\t'$, $\x'$ etc.\ instead of $\t^T$, $x^T$ etc.\ to denote the transpose of $\t$, $\x$ etc.}
$$
\|T\|^2 = \t' B \t,
$$
where $\t \in \Z^{(k-1)^2}$ is a vector with coordinates $t_{i,j}$, $1 \le i,j \le k-1$ (in some fixed linear order),
and $B$ is a symmetric positive definite matrix. With this, we have
$$
\tE[\|A-U\|^2] = \frac{1}{\tZ} \sum_{\t\in\Z^{(k-1)^2}} \t' B \t \, \exp\left\{-\frac{1}{2\sigma_{k,M}^2} \, \t' B \t\right\},
$$
where $\tZ$ can now be written as
$$
\tZ = \sum_{\t \in \Z^{(k-1)^2}} \exp\left\{-\frac{1}{2\sigma_{k,M}^2} \, \t' B \t\right\}.
$$
Thus, we see that $\tE[\|A-U\|^2]$ is equal to the expected value of $\t' B \t$, where $\t$ is a random vector in $\Z^{(k-1)^2}$ distributed according to the discrete Gaussian measure $\frac{1}{\tZ} \exp\left\{-\frac{1}{2\sigma_{k,M}^2} \, \t' B \t\right\}$. The lemma below is a direct consequence of Proposition~\ref{prop:var} in Appendix~F.

\begin{lemma}
We have
$$
\tE[\|A-U\|^2] = (k-1)^2 \sigma_{k,M}^2 + O\bigl(\exp(-c_k M)\bigr),
$$
where $c_k$ is a positive constant depending only on $k$.
\label{lem:Qtilde_var}
\end{lemma}

Equation~\eqref{eq:E[A-U]} and Lemmas~\ref{lem:E_and_tE} and \ref{lem:Qtilde_var} suffice to prove that 
$$
\Var_{k,M}(a_{1,1}) = \frac{(k-1)^2}{k^2} \, \sigma_{k,M}^2 + o(M),
$$
from which we obtain the second statement of Proposition~\ref{prop:VarkM}(b).


\section{Concluding Remarks} \label{sec:conclusion}

The two main results of this paper, namely, Theorems~\ref{thm:thres} and \ref{thm:thresB}, encapsulate an abrupt transition in behaviour of the uniform distribution $U_k$ in relation to the mapping $\beta_{k,M}^{\psibf}:\Pi_{[k]} \to \R_+$ defined by $\p \mapsto \perm_{B,M}(\Th(\psibf;\p))$. While the former theorem concerns the $M=1$ case, the latter theorem holds for all sufficiently large $M$. Our proofs of these theorems involve picking an arbitrary unit vector $\xibf \in \R^k$ such that $\p(t) = U_k + t \xibf$ lies within $\Pi_{[k]}$ for all $t$ within a sufficiently small interval around $0$, and analyzing the first and second derivatives at $t = 0$ of the function $G_{k,M}(t) = \beta_{k,M}^{\psibf}(\p(t))$. 

Unfortunately, we have not been able to successfully extend this proof technique to deduce an analogous result for the mapping $\beta_{k,\infty}^{\psibf}: \Pi_{[k]} \to [0,1]$ defined by $\p \mapsto \perm_B(\Th(\psibf;\p))$. The main technical obstacle here is the fact that the corresponding function $G_{k,\infty}(t) = \beta_{k,\infty}^{\psibf}(\p(t))$ may not be differentiable. (We do know that this function is continuous, since $\perm_B(\psibf;\p)$ was noted to be a continuous function of $\p$ in the paragraph following \eqref{def:BPML}.) Note that Vontobel's identity \eqref{eq:Vontobel} only allows us to say that $\limsup_{M\to\infty} G_{k,M}(t) = G_{k,\infty}(t)$ pointwise in $t$. By itself, this is insufficient to claim the differentiability of $G_{k,\infty}$. 

It would be possible to prove the desired result if we could show that the functions $G_{k,M}$ satisfy the following two conditions within a sufficiently small interval $I$ about $0$: (i)~$\lim_{M\to\infty} G_{k,M}(t)$ exists pointwise on $I$, and (ii)~the first three derivatives of $G_{k,M}$ are uniformly bounded on $I$. Indeed, using the properties of equicontinuous and uniformly convergent functions (for example, Theorems~7.17 and 7.25 in \cite{Rudin}), we can then deduce that $G_{k,\infty}$ is twice-differentiable in the interior of $I$. Moreover, we would have $\lim_{M \to \infty} G_{k,M}(t) = G_{k,\infty}(t)$, $\lim_{M \to \infty} G'_{k,M}(t) = G'_{k,\infty}(t)$, and $\lim_{M \to \infty} G''_{k,M}(t) = G''_{k,\infty}(t)$ for all $t$ in the interior of $I$. In particular, these hold at $t=0$. From this, we could conclude, via Lemma~\ref{lem:G} and Corollary~\ref{cor:signG''}, that the mapping $\beta_{k,\infty}^{\psibf}$ admits a phase transition at the same threshold as in Theorem~\ref{thm:thresB}.

We would unhesitatingly conjecture that condition~(i) above is true, but we are less sure about condition~(ii). It is actually not difficult to show that $G_{k,M}'$ is uniformly bounded within some small interval $I$, but the uniform boundedness of the second and third derivatives presents difficulties. It is in fact entirely possible that this approach will not work, as it may be the case that $G_{k,\infty}$ is not differentiable. However, we do believe that the ultimate result is still true: there is sufficient numerical evidence in favour of there being a phase transition at the threshold $\Upsilon_B(\psibf)$, defined in Theorem~\ref{thm:thresB}, for the uniform distribution in the Bethe PML problem.


\section*{Appendix~A: Proof of Proposition \ref{prop:permBM}}

From Definition~\ref{def:permBM}, we see that $\perm_{B,M}(\Th)$ involves the permanents of $kM \times kM$ matrices $\Th \odot \Lambda$ as in \eqref{M-lift}. Each such permanent involves a sum over permutations $\pi: [kM] \to [kM]$. We identify a permutation $\pi: [kM] \to [kM]$ with the permutation matrix in $\cP_{kM}$ whose $(s,t)$th entry is a $1$ iff $\pi(s) = t$; in a slight abuse of notation, we let $\pi$ denote this permutation matrix as well. We will find it convenient to view any matrix $\pi \in \cP_{kM}$ as a block matrix of the form 
\begin{equation}
\pi = \left(
\begin{array}{cccc}
\pi^{(1,1)} & \pi^{(1,2)} & \cdots & \pi^{(1,k)} \\
\pi^{(2,1)} & \pi^{(2,2)} & \cdots & \pi^{(2,k)} \\
\vdots & \vdots & \ddots & \vdots \\
\pi^{(k,1)} & \pi^{(k,2)} & \cdots & \pi^{(k,k)} \\
\end{array}
\right)
\label{pi_blocks}
\end{equation}
where for $1 \le i,j \le k$, the $(i,j)$th block $\pi^{(i,j)}$ is the $M \times M$ submatrix of $\pi$ located at the intersection of the rows and columns indexed by $(i-1)M+1,\ldots,iM$ and $(j-1)M+1,\ldots,jM$, respectively. Let $a_{i,j}(\pi)$ denote the number of $1$s in $\pi^{(i,j)}$.  

Given two $0/1$-matrices $\mathsf{P} = (p_{i,j})$ and $\mathsf{Q} = (q_{i,j})$ of the same size, we write $\mathsf{P} \leq \mathsf{Q}$ if for all $i,j$, we have $p_{i,j} \le q_{i,j}$ (or equivalently, $p_{i,j} = 1 \Longrightarrow q_{i,j} = 1$). Using the newly introduced notation, we observe that 
\begin{equation}
\perm(\Th \odot \Lambda) 
 = \sum_{\pi \in \cP_{kM}: \pi \le \Lambda} \prod_{(i,j) \in [k]^2} \theta_{i,j}^{a_{i,j}(\pi)}
\label{permodot}
\end{equation}
since permutations $\pi \not\le \Lambda$ contribute only $0$s to the permanent. Hence,
\begin{align}
\bigl\langle\perm(\Th\odot\Lambda)\bigr\rangle
& = (M!)^{-k^2} \sum_{\Lambda \in \cP_M^{k \times k}} \sum_{\pi \in \cP_{kM}: \pi \leq \Lambda} \prod_{(i,j) \in [k]^2}  \theta_{i,j}^{a_{i,j}(\pi)} \notag \\
& = (M!)^{-k^2} \sum_{\pi \in \cP_{kM}} \sum_{\Lambda \in \cP_{M}^{k \times k}: \pi \leq \Lambda} \prod_{(i,j) \in [k]^2}  \theta_{i,j}^{a_{i,j}(\pi)}. 
\label{avgperm1}
\end{align}
Now, for a given $\pi \in \cP_{kM}$, the number of matrices $\Lambda$ as in \eqref{eq:Lambda} such that $\pi \leq \Lambda$ is equal to $\prod_{(i,j) \in [k]^2}  (M-a_{i,j}(\pi))!$. This is because, for each $(i,j)$, $\pi^{(i,j)}$ determines the positions of $a_{i,j}(\pi)$ $1$s in $P^{(i,j)}$, and the positions of the remaining $1$s in $P^{(i,j)}$ can be chosen in $(M-a_{i,j}(\pi))!$ ways to make $P^{(i,j)}$ a permutation matrix. Therefore, carrying on from \eqref{avgperm1}, we have
\begin{equation}
\bigl\langle\perm(\Th\odot\Lambda)\bigr\rangle 
 = (M!)^{-k^2} \sum_{\pi \in \cP_{kM}} \prod_{(i,j) \in [k]^2}  (M-a_{i,j}(\pi))! \, \theta_{i,j}^{a_{i,j}(\pi)}. 
\label{avgperm2}
\end{equation}

For any $\pi \in \cP_{kM}$, the $k \times k$ matrix $A(\pi) := (a_{i,j}(\pi))$ is a non-negative integer matrix whose row and column sums are all equal to $M$. Recall from the statement of Proposition~\ref{prop:permBM} that $\cA_{k,M}$ denotes the set of all such $k \times k$ matrices. Thus, we can write \eqref{avgperm2} as 
\begin{equation}
\bigl\langle\perm(\Th\odot\Lambda)\bigr\rangle 
 = (M!)^{-k^2} \sum_{A = (a_{i,j}) \in \cA_{k,M}} |\cP_{kM}(A)| \prod_{(i,j) \in [k]^2} (M-a_{i,j})! \, \theta_{i,j}^{a_{i,j}} 
\label{avgperm3}
\end{equation}
where $\cP_{kM}(A) := \{\pi \in \cP_{kM}: A(\pi) = A\}$. Therefore, the proof of Proposition~\ref{prop:permBM} would be complete once we prove the following lemma.

\begin{lemma}
For $A = (a_{i,j}) \in \cA_{k,M}$, we have 
$$|\cP_{kM}(A)| = \frac{(M!)^{2k}}{\prod_{(i,j) \in [k]^2}  (a_{i,j})!}.$$
\label{lem:PkMA}
\end{lemma}
\begin{proof}
Given a matrix $A = (a_{i,j}) \in \cA_{k,M}$, we can construct permutation matrices $\pi \in \cP_{kM}$ such that $A(\pi) = A$ by following the three steps described below. Our description views a $kM \times kM$ matrix $\pi$ as a $k \times k$ block matrix as in \eqref{pi_blocks}. 
\begin{enumerate}
\item Fix an $i \in [k]$. For each $j \in [k]$, pick $a_{i,j}$ rows of $\pi^{(i,j)}$ within which to place $1$s. Since $\pi$ cannot have two $1$s in the same row, the number of ways in which these $a_{i,j}$ rows, $j = 1,\ldots,k$, can be picked is the multinomial coefficient $\binom{M}{a_{i,1},\ldots,a_{i,k}}$. Then, letting $i$ range over $[k]$, we see that the number of ways in which rows of $\pi$ can be so chosen is $\prod_i \binom{M}{a_{i,1},\ldots,a_{i,k}}$.
\item Fix a $j \in [k]$. For each $i \in [k]$, pick $a_{i,j}$ columns of $\pi^{(i,j)}$ within which to place $1$s. By a similar argument as above, this can be done in $\binom{M}{a_{1,j},\ldots,a_{k,j}}$ ways. So, letting $j$ range over $[k]$, we see that the number of ways in which columns of $\pi$ can be so chosen is $\prod_j \binom{M}{a_{1,j},\ldots,a_{k,j}}$.
\item Fix a pair $(i,j) \in [k] \times [k]$, and consider the submatrix of $\pi^{(i,j)}$ formed by the points of intersection of the $a_{i,j}$ rows and columns chosen in the first two steps. For $\pi$ to be a permutation matrix, this submatrix should be a permutation matrix as well. Hence, there are $(a_{i,j})!$ ways of placing $0$s and $1$s within this submatrix. All other entries of $\pi^{(i,j)}$ must be $0$s. Letting $(i,j)$ range over $[k] \times [k]$, we determine the number of possible choices in this step to be $\prod_{i,j} (a_{i,j})!$.
\end{enumerate}
Thus, putting together the counts in the three steps, we obtain 
$$
|\cP_{kM}(A)| = \left[\prod_i \binom{M}{a_{i,1},\ldots,a_{i,k}} \right] \Biggl[\prod_j \binom{M}{a_{1,j},\ldots,a_{k,j}}\Biggr] \biggl[\prod_{i,j}  (a_{i,j})!\biggr],
$$
which simplifies to the expression in the statement of the lemma. 
\end{proof}


\section*{Appendix~B: Proof of Lemma~\ref{lem:G}}

We introduce some convenient notation to be used in the proof. For $A = (a_{i,j}) \in \cA_{k,M}$, define $\gm_A(t) = \prod_{i,j} (p_i+t\xi_i)^{\mu_ja_{i,j}}$ and $\gm(t) = \sum_{A \in \cA_{k,M}} w(A) \gm_A(t)$. We will need the values of $\gm(0)$, $\gm'(0)$ and $\gm''(0)$, for which we need to compute $\gm_A(0)$, $\gm_A'(0)$ and $\gm_A''(0)$. 

Determining $\gm_A(0)$ is easy: since $p_i = \frac1k$ for all $i \in [k]$, we have
$$
\gm_A(0) = \prod_{i,j} (\frac{1}{k})^{\mu_j a_{i,j}} = k^{-\sum_{i,j} \mu_j a_{i,j}} = k^{-Mn},
$$
the last equality using $\sum_i a_{i,j} = M$ and $\sum_j \mu_j = n$. Hence,
\begin{equation}
\gm(0) = \sum_A w(A) \gm_A(0) = k^{-Mn} Z_{k,M}.
\label{gm0}
\end{equation}

Similarly, taking the derivative of $\gm_A(t)$, it is straightforward to show that
$$
\gm_A'(0) = k^{1-Mn} \sum_{i,j} \xi_i \mu_j a_{i,j},
$$
and hence,
$$
\gm'(0) = \sum_A w(A) \gm_A'(0) = k^{1-Mn} \sum_{i,j} \xi_i \mu_j \sum_A a_{i,j} w(A).
$$
Now, $\sum_A a_{i,j} w(A) = Z_{k,M} \E[a_{i,j}]$, where $\E[\cdot]$ denotes expectation taken with respect to a random matrix $A \in \cA_{k,M}$ distributed according to $Q_{k,M}$. 
By Lemma~\ref{lem:Eaij}, $\E[a_{i,j}]$ is a constant independent of $(i,j)$, and hence,
\begin{equation}
\gm'(0) = k^{1-Mn} Z_{k,M} (\text{const.}) \biggl(\sum_i\xi_i\biggr) \biggl(\sum_j \mu_j\biggr) = 0,
\label{gm'0}
\end{equation}
since $\sum_i \xi_i = 0$ by choice of the direction vector $\xibf$.

Calculation of the second derivative $\gm''(0)$ requires a lot more work. To start with, routine differentiation yields
$$
\gm_A''(0) = k^{2-Mn}\left[ \biggl(\sum_{i,j} \xi_i \mu_j a_{i,j} \biggr)^2 - \sum_{i,j} \xi_i^2 \mu_j a_{i,j} \right],
$$
which we can plug into 
$$
\gm''(0) = \sum_A w(A) \gm_A''(0) = Z_{k,M} \E[\gm_A''(0)]
$$
to get
\begin{align}
\gm''(0) &= Z_{k,M} k^{2-Mn} \left( \E\biggl[\biggl(\sum_{i,j} \xi_i \mu_j a_{i,j} \biggr)^2\biggr] 
- \sum_{i,j} \xi_i^2 \mu_j \E[a_{i,j}] \right) \notag \\
& = Z_{k,M} k^{2-Mn} \left( \E\biggl[\biggl(\sum_{i,j} \xi_i \mu_j a_{i,j} \biggr)^2\biggr] - \frac{Mn}{k}\right).
\label{gm''0_eq1}
\end{align}
For the second equality above, we used Lemma~\ref{lem:Eaij} and $\sum_{i,j} \xi_i^2 \mu_j = (\sum_i \xi_i^2) (\sum_j \mu_j) = n$, since ${\|\xibf\|}_2 = 1$.

To determine $\E\bigl[\bigl(\sum_{i,j} \xi_i \mu_j a_{i,j} \bigr)^2\bigr]$, we write
\begin{equation}
\E\biggl[\biggl(\sum_{i,j} \xi_i \mu_j a_{i,j} \biggr)^2\biggr] = 
 \sum_{i,j} \xi_i^2 \mu_j^2 \E[a_{i,j}^2] + \sum_{\substack{(i,j), (i',j'): \\ (i,j) \ne (i',j')}} \xi_i\xi_{i'} \mu_j \mu_{j'} \E[a_{i,j}a_{i',j'}]
\label{Esumsq}
\end{equation}

As argued for $\E[a_{i,j}]$, the expected value $\E[a_{i,j}^2]$ is also a constant independent of $(i,j)$. With this, the first term on the right-hand side (RHS) of \eqref{Esumsq} can be expressed as 
\begin{equation}
\sum_{i,j} \xi_i^2 \mu_j^2 \E[a_{i,j}^2] = \E[a_{1,1}^2] \sum_j \mu_j^2,
\label{term1}
\end{equation}
using ${\|\xibf\|}_2 = 1$.

Turning our attention to the second term on the RHS of \eqref{Esumsq}, we note that for $(i,j) \ne (i',j')$, by virtue of the invariance of $Q_{k,M}$ with respect to row and column permutations,
\begin{equation}
\E[a_{i,j} a_{i',j'}] = 
\begin{cases}
\E[a_{1,1}a_{2,2}] & \text{ if } i \ne i' \text{ and }  j \ne j' \\
\E[a_{1,1}a_{1,2}] & \text{ otherwise.}
\end{cases}
\label{Eiji'j'}
\end{equation}
Now,
\begin{align}
\E[a_{1,1}a_{1,2}] &= \frac{1}{k-1} \sum_{j=2}^k \E[a_{1,1}a_{1,j}] \ = \ \frac{1}{k-1} \E\biggl[a_{1,1} \sum_{j=2}^k a_{1,j} \biggr]
\notag \\
&= \frac{1}{k-1} \E\bigl[a_{1,1}(M-a_{1,1}) \bigr] \ = \ \frac{1}{k-1} \biggl[\frac{M^2}{k} - \E[a_{1,1}^2] \biggr],
\label{E1112}
\end{align}
where we used Lemma~\ref{lem:Eaij} to get the last equality. By a similar argument,
\begin{align}
\E[a_{1,1}a_{2,2}] &= \frac{1}{k-1} \E\biggl[a_{1,1} \sum_{i=2}^k a_{i,2}\biggr] \ = \ \frac{1}{k-1} \E\bigl[a_{1,1}(M-a_{1,2}) \bigr] \notag \\
& = \frac{1}{k-1} \biggl[\frac{M^2}{k} - \E[a_{1,1} a_{1,2}] \biggr] 
  \ = \ \frac{1}{(k-1)^2} \biggl[\frac{k-2}{k} M^2 + \E[a_{1,1}^2] \biggr],
\label{E1122}
\end{align}
this time using \eqref{E1112} to get the last equality. Plugging \eqref{term1}--\eqref{E1122} into \eqref{Esumsq}, and then performing some careful book-keeping, we eventually obtain
$$
\E\biggl[\biggl(\sum_{i,j} \xi_i \mu_j a_{i,j} \biggr)^2\biggr] = \frac{k}{(k-1)^2} \Var_{k,M}(a_{1,1}) \biggl[k \sum_j \mu_j^2 - n^2\biggr],
$$
and hence,
\begin{equation}
\gm''(0) = Z_{k,M} k^{2-Mn} \left( \frac{k}{(k-1)^2} \Var_{k,M}(a_{1,1}) \biggl[k \sum_j \mu_j^2 - n^2\biggr] - \frac{Mn}{k}\right). 
\label{gm''0_eq2}
\end{equation}

\medskip

Finally, by Proposition~\ref{prop:permBM}, we have 
$$
G_{k,M}(t) = \perm_{B,M}(\Th(\psibf;\p(t))) = [(M!)^{2k-k^2}]^{\frac1M} \gm(t)^{\frac1M}.
$$
Taking the derivative with respect to $t$ and setting $t=0$, we obtain 
$$
G_{k,M}'(0) = [(M!)^{2k-k^2}]^{\frac1M} \frac{1}{M} \gm(0)^{\frac1M -1} \gm'(0) = 0,
$$
since $\gm'(0) = 0$ --- see \eqref{gm'0}.

 Similarly, calculating the second derivative at $t=0$, we get 
$$
G_{k,M}''(0) =  [(M!)^{2k-k^2}]^{\frac1M} \frac{1}{M} \, \gm(0)^{\frac1M} \, \frac{\gm''(0)}{\gm(0)}.
$$
Plugging in the expressions for $\gm(0)$ and $\gm''(0)$ given in \eqref{gm0} and \eqref{gm''0_eq2}, respectively, we obtain the expression for $G_{k,M}''(0)$ recorded in the statement of Lemma~\ref{lem:G}.


\section*{Appendix~C: Proof of Lemma~\ref{lem:roots}}

Part~(1) of the lemma is obvious, so we concern ourselves with part~(2). Here, there are two claims that need proof:

\emph{Claim~\ref{lem:roots}.1}. The discriminant $D$ is strictly negative iff $\Upsilon < \frac{\sqrt{n}+1}{\sqrt{n}-1}$.

\emph{Claim~\ref{lem:roots}.2}. When $D \ge 0$, so that real roots $\rho_1$ and $\rho_2$ exist, we have $1 < \rho_1 \le 2$ and $\Upsilon - 1 \le \rho_2 < \Upsilon$.

\begin{proof}[Proof of Claim~\ref{lem:roots}.1]
Note that $D = (n^2+2n-\cU)^2 - 4n^3 < 0$ iff $|n^2+2n-\cU| < 2n^{3/2}$. We may remove the absolute value in the latter condition since $n^2 > \cU$. Thus, $D < 0$ iff $n^2+2n-\cU < 2n^{3/2}$, which upon some re-arrangement becomes $\cU - n > (n-\sqrt{n})^2$. Thus, recalling that $\Upsilon = \frac{n^2 - n}{\cU - n}$, we see that $D < 0$ is equivalent to 
$$
\Upsilon < \frac{n^2 - n}{(n-\sqrt{n})^2}.
$$
Upon cancelling the common factor $\sqrt{n}(n-\sqrt{n})$, the right-hand side simplifies to $\frac{\sqrt{n}+1}{\sqrt{n}-1}$. \end{proof}

\begin{proof}[Proof of Claim~\ref{lem:roots}.2]
We start by writing $\rho_1 = 1 + \frac{n^2 - \cU - \sqrt{D}}{2(\cU-n)}$. To show that $\rho_1 > 1$, it suffices to show that $n^2 - \cU > \sqrt{D}$, or equivalently, $(n^2-\cU)^2 > D$. Routine algebra shows that 
$(n^2-\cU)^2 - D = 4n\cU$, which is of course positive.

For the rest of this proof, it will be convenient to define $a = \cU-n$, $b = \cU+n^2-2n$ and $c=n^2$, so that $q(x) = ax^2-bx+c$. With this, $\rho_1 = \frac{b-\sqrt{b^2-4ac}}{2a}$. 

We now give a proof for $\rho_1 \le 2$. Suppose that $\rho_1 \ge 2$. We would then have $b-4a \ge \sqrt{b^2-4ac}$. This yields two inequalities that must necessarily be satisfied: $b - 4a \ge 0$ and $(b-4a)^2 \ge \sqrt{b^2-4ac}$. The latter inequality simplifies to $4a-2b+c \ge 0$, so that the two inequalities that must be satisfied are:
\begin{equation}
b - 4a \ge 0 \ \ \text{ and } \ \ 4a-2b+c \ge 0
\label{eq:abc}
\end{equation}
Plugging in the expressions for $a$, $b$ and $c$, we obtain $n^2 + 2n \ge 3\cU$ and $2 \cU \ge n^2$. Combining these inequalities, we get 
\begin{equation}
\frac12 n^2 \le \cU \le \frac13(n^2+2n),
\label{U_ineqs}
\end{equation} 
and hence, $\frac12 n^2 \le \frac13(n^2+2n)$. Upon re-arrangement, this becomes $n^2 - 4n \le 0$, which yields $0 \le n \le 4$. 

As we do not consider pattern lengths $n < 2$, we must deal with $n=2,3,4$. Plugging these values of $n$ into \eqref{U_ineqs}, we obtain $(n,\cU) = (2,2)$, $(3,5)$ and $(5,8)$ as the only valid solutions. The assumption of part~(2) of the lemma is that $\cU > n$, so we cannot have $(n,\cU) = (2,2)$. Also, $(n,\cU) = (3,5)$ is not possible as this yields a negative discriminant. We are thus forced to conclude that if $\rho_2 \ge 2$, then $(n,\cU) = (4,8)$. Indeed, in this case, we have $q(x) = 4x^2 - 16x + 16$, so that $\rho_1 = \rho_2 = 2$. We have thus proved that $\rho_1 \le 2$ always holds, and in fact, it holds with equality iff $(n,\cU) = (4,8)$, i.e., $\psibf = 1122$.

\medskip

To prove that $\Upsilon-1 \le \rho_2 < \Upsilon$, consider the difference $\rho_2 - \Upsilon$. It may be verified that this difference can be expressed as 
$$
\rho_2 - \Upsilon = \frac{1}{2(\cU-n)}\left[-(n^2-\cU) + \sqrt{(n^2-\cU)^2 - 4n(\cU-n)} \right],
$$
which is obviously strictly negative, given that $\cU > n$ and $n^2 > \cU$. Thus, $\rho_2 < \Upsilon$.

Now, suppose $\rho_2 - \Upsilon \le -1$. Then, using the expression given above for $\rho_2 - \Upsilon \le -1$, we must have $\sqrt{(n^2-\cU)^2 - 4n(\cU-n)} \le n^2 + 2n - 3\cU$. This yields two inequalities to be satisfied: $n^2 + 2n - 3\cU \ge 0$ and ${(n^2-\cU)^2 - 4n(\cU-n)} \le (n^2 + 2n - 3\cU)^2$. The latter inequality can be manipulated into the following equivalent form: $4(\cU-n)(2\cU-n^2) \ge 0$. Since $\cU > n$, this inequality is satisfied iff $2 \cU \ge n^2$. Thus, $\rho_2 - \Upsilon \le -1$ only if the inequalities in \eqref{U_ineqs} hold. As argued earlier, these inequalities are satisfied only if $(n,\cU) = (4,8)$, in which case it may be verified that $\rho_2 - \Upsilon = -1$. This proves that $\rho_2 \ge \Upsilon -1$ always holds, and again, it holds with equality iff $(n,\cU) = (4,8)$, i.e., $\psibf = 1122$.
\end{proof}

This completes the proof of Lemma~\ref{lem:roots}.


\section*{Appendix~D: Proof of Lemma~\ref{lem:gauss}}

Throughout the proof, we fix $k \ge 3$. We introduce some convenient notation: $\rho := M/k$, $\ve_{i,j}  :=  u_{i,j} - \rho$ and $f(x) := x(M+1-x)$. Now, consider the ratio $w(A)/w(U) = w(U+T)/w(U)$:
\begin{align}
\frac{w(U+T)}{w(U)} &= \prod_{(i,j): t_{i,j} \ne  0} \frac{\bigl(M-(u_{i,j}+t_{i,j})\bigr)!}{(M-u_{i,j})!} \frac{u_{i,j}!}{(u_{i,j}+t_{i,j})!} \notag \\
& = \prod_{(i,j): t_{i,j} < 0} f(u_{i,j}-|t_{i,j}|+1) \cdots f(u_{i,j}) \prod_{(i,j):t_{i,j} > 0} \frac{1}{f(u_{i,j} + 1) \cdots f(u_{i,j}+t_{i,j})} \notag \\
& = \prod_{(i,j): t_{i,j} < 0} \frac{f(u_{i,j}-|t_{i,j}|+1) \cdots f(u_{i,j})}{f(\rho)^{|t_{i,j}|}} \prod_{(i,j):t_{i,j} > 0} \frac{f(\rho)^{t_{i,j}}}{f(u_{i,j} + 1) \cdots f(u_{i,j}+t_{i,j})} \notag 
\end{align}
The last equality above holds because $\sum_{i,j} t_{i,j} = 0$. Thus,
\begin{equation}
\log \frac{w(U+T)}{w(U)} 
   = \sum_{(i,j):t_{i,j} < 0} \sum_{\ell = -|t_{i,j}|+1}^{0} \log \frac{f(u_{i,j}+\ell)}{f(\rho)}  
      - \sum_{(i,j):t_{i,j} > 0} \sum_{\ell = 1}^{t_{i,j}} \log \frac{f(u_{i,j}+\ell)}{f(\rho)} 
\label{log_wratio}
\end{equation}

Thus, we need estimates for the summands $\log \frac{f(u_{i,j} + \ell)}{f(\rho)}$. Observe that, since $u_{i,j}+t_{i,j} = a_{i,j} \ge 0$, the integers $\ell$ that appear in \eqref{log_wratio} all satisfy $u_{i,j} + \ell \ge 1$. We will make use of this observation a little later.

We first derive useful estimates for the ratios $\frac{f(u_{i,j} + \ell)}{f(\rho)}$. Note that $f(x_2)-f(x_1) = (x_2-x_1)(M+1-(x_1+x_2))$. Using this and the fact that $M = k\rho$, we can write, for $\ell \in \Z$:
\begin{align}
\frac{f(u_{i,j} + \ell)}{f(\rho)} & = 1 + \frac{f(u_{i,j} + \ell) - f(\rho)}{f(\rho)}  \notag \\
& = 1 + \frac{(\ell+\ve_{i,j})((k-2)\rho - \ell - \ve_{i,j} +  1)}{\rho((k-1)\rho + 1)}  \notag \\
& = 1 + \left(\frac{\ell+\ve_{i,j}}{\rho}\right) \left(\frac{k-2}{k-1}\right) (1+ \gamma(\ell)), \label{fratio_eq1}
\end{align}
where $1+\gamma(\ell) = \bigl(1 - \frac{\ell+\ve_{i,j}-1}{(k-2)\rho}\bigr)\bigl(1+ \frac{1}{(k-1)\rho}\bigr)^{-1}$. Observe that 
\begin{equation}
1 + \gamma(\ell) \le 1 - \frac{\ell + \ve_{i,j} - 1}{(k-2)\rho} \le 1+ \frac{|\ell|+2}{(k-2)\rho},
\label{gamma_upbnd}
\end{equation}
using $\ve_{i,j} \in (-1,1)$. On the other hand, using $\bigl(1+ \frac{1}{(k-1)\rho}\bigr)^{-1}  \ge 1-\frac{1}{(k-1)\rho}$, we also have
\begin{align}
1 + \gamma(\ell) & \ge \biggl(1 - \frac{\ell+\ve_{i,j}-1}{(k-2)\rho}\biggr)\biggl(1 - \frac{1}{(k-1)\rho}\biggr) \notag \\
& = 1 - \frac{\ell+\ve_{i,j}}{(k-2)\rho}   + \frac{\rho + \ell + \ve_{i,j} - 1}{(k-1)(k-2)\rho^2} \notag \\
& = 1 - \frac{\ell+\ve_{i,j}}{(k-2)\rho}   + \frac{\ell + u_{i,j} - 1}{(k-1)(k-2)\rho^2} \notag \\
& \ge 1 - \frac{\ell+\ve_{i,j}}{(k-2)\rho} \label{gamma_lobnd}
\end{align}
since, as observed earlier, $\ell + u_{i,j} \ge 1$ for the integers $\ell$ that appear in \eqref{log_wratio}. It should also be pointed out that the first inequality above requires $1 - \frac{\ell+\ve_{i,j}-1}{(k-2)\rho}$ to be non-negative. If $1 - \frac{\ell+\ve_{i,j}-1}{(k-2)\rho} < 0$, then $1 + \gamma(\ell)$ is still lower bounded by \eqref{gamma_lobnd}, since we now have $1 + \gamma(\ell) =  \bigl(1 - \frac{\ell+\ve_{i,j}-1}{(k-2)\rho}\bigr)\bigl(1+ \frac{1}{(k-1)\rho}\bigr)^{-1}
 \ge 1 - \frac{\ell+\ve_{i,j}-1}{(k-2)\rho}.$

Thus, from \eqref{fratio_eq1}--\eqref{gamma_lobnd}, we obtain for any integer $\ell$ occurring in \eqref{log_wratio},
\begin{equation}
\frac{f(u_{i,j} + \ell)}{f(\rho)} = 1 + \zeta(\ell)(1+\gamma(\ell)),
\label{fratio_eq2}
\end{equation}
where $\zeta(\ell) = \left(\frac{k-2}{k-1}\right) \left(\frac{\ell+\ve_{i,j}}{\rho}\right)$, and $|\gamma(\ell)| \le \frac{|\ell| + 2}{(k-2)\rho}$. Observe that for $|\ell| \ge 1$, $|\zeta(\ell)(1+\gamma(\ell))| \le \frac{2|\ell|}{\rho}\left(1+\frac{3|\ell|}{\rho}\right),$ which is at most $\frac12$ for $\frac{|\ell|}{\rho} \le \frac19$. (Indeed, it is easy to check that $2x(1+3x) \le \frac12$ for $|x| \le \frac{1}{12}(\sqrt{15}-2) = 0.1560\ldots$.)  Also, verify that $|\zeta(0) (1+\gamma(0)) |\le  \frac{1}{\rho} \left(1 + \frac{2}{(k-2)\rho}\right) \le \frac{1}{\rho}\left(1 + \frac{2}{\rho}\right)$, which is at most $\frac{3}{8}$ for $\rho \ge 4$.  Henceforth, we assume $\rho \ge 4$ and $\frac{|\ell|}{\rho} \le \frac19$. 

From \eqref{fratio_eq2}, via the inequality $x-x^2 \le \log(1+x) \le x$, valid for $|x| \le \frac12$, we obtain
$$
-{\bigl[\zeta(\ell)(1+\gamma(\ell))\bigr]}^2 \le \log \frac{f(u_{i,j} + \ell)}{f(\rho)} - \zeta(\ell)(1+\gamma(\ell)) \le 0,
$$
and hence,
$$
\zeta(\ell)\gamma(\ell)-{\bigl[\zeta(\ell)(1+\gamma(\ell))\bigr]}^2 \le \log \frac{f(u_{i,j} + \ell)}{f(\rho)} - \zeta(\ell) \le \zeta(\ell)\gamma(\ell).
$$
It follows that 
\begin{equation}
\left| \log \frac{f(u_{i,j} + \ell)}{f(\rho)} - \zeta(\ell)  \right|  
   \le |\zeta(\ell)\gamma(\ell)| + {\bigl[\zeta(\ell)(1+\gamma(\ell))\bigr]}^2
\label{logf_bnds_eq1}
\end{equation}
Now, $|\zeta(\ell)\gamma(\ell)| \le \frac{1}{k-1} \frac{(|\ell| + 1)(|\ell|+2)}{\rho^2} \le \frac{2}{k-1} \left(\frac{|\ell|+1}{\rho}\right)^2 \le \left(\frac{|\ell|+1}{\rho}\right)^2$. Furthermore, $|1+\gamma(\ell)| \le 1+|\gamma(\ell)| \le 1+\frac{|\ell|+2}{\rho} \le 1+\frac{1}{9} + \frac{1}{4} = \frac{29}{18}$, as we have assumed $\frac{|\ell|}{\rho} \le \frac19$ and $\rho \ge 4$. From this, we get $|(\zeta(\ell)(1+\gamma(\ell))| \le \frac{29}{18} \frac{|\ell|+1}{\rho}$. Plugging these estimates into \eqref{logf_bnds_eq1}, we obtain
\begin{equation}
\left| \log \frac{f(u_{i,j} + \ell)}{f(\rho)} - \zeta(\ell)  \right|  
   \le \left(1 + \frac{29^2}{18^2}\right) \left(\frac{|\ell|+1}{\rho}\right)^2
   \le 4 \left(\frac{|\ell|+1}{\rho}\right)^2.
\label{logf_bnds_eq2}
\end{equation}

From \eqref{logf_bnds_eq2}, we can deduce estimates for sums of the form $\sum_\ell \log \frac{f(u_{i,j} + \ell)}{f(\rho)}$, which we can use in \eqref{log_wratio}. Indeed, for an integer $t < 0$, with $|t| \le \frac19 \rho$, we have
\begin{equation*}
\left| \sum_{\ell = -|t|+1}^0 \log \frac{f(u_{i,j} + \ell)}{f(\rho)} - \sum_{\ell = -|t|+1}^0 \zeta(\ell)  \right|  
 \le \frac{4}{\rho^2} \sum_{\ell = -|t|+1}^0 (|\ell|+1)^2,
\end{equation*}
which yields
$$
\left| \sum_{\ell = -|t|+1}^0 \log \frac{f(u_{i,j} + \ell)}{f(\rho)} + \frac12 \biggl(\frac{k-2}{k-1}\biggr) \biggl(\frac{1}{\rho}\biggr) \bigl[|t|^2+ |t|(2\ve_{i,j}-1)\bigr]  \right|  
 \le \frac{2}{3\rho^2}\bigl[|t|(|t|+1)(2|t|+1) \bigr].
$$
The above bound can be brought into the following looser but simpler form:
\begin{equation}
\left| \sum_{\ell = -|t|+1}^0 \log \frac{f(u_{i,j} + \ell)}{f(\rho)} + \frac12 \biggl(\frac{k-2}{k-1}\biggr) \biggl(\frac{t^2}{\rho}\biggr) \right| 
  \le  \frac{4|t|^3}{\rho^2} + \frac{3|t|}{2\rho}.
\label{neg_t_bound}
\end{equation}

Similarly, for an integer $0 < t \le \frac19 \rho$, we can obtain
\begin{equation}
\left| \sum_{\ell = 1}^t \log \frac{f(u_{i,j} + \ell)}{f(\rho)} - \frac12 \biggl(\frac{k-2}{k-1}\biggr) \biggl(\frac{t^2}{\rho}\biggr) \right| 
  \le  \frac{4(t+1)^3}{\rho^2} + \frac{3t}{2\rho}.
\label{pos_t_bound}
\end{equation}
Lemma~\ref{lem:gauss} readily follows from \eqref{log_wratio}, \eqref{neg_t_bound} and \eqref{pos_t_bound}.


\section*{Appendix~E: Proof of Lemma~\ref{lem:E_and_tE}}

As in Lemma~\ref{lem:gauss} and Appendix~D, we set $\rho = \frac{M}{k}$. We will use the notation introduced after the statement of Lemma~\ref{lem:E_and_tE} in Section~\ref{sec:varkM}. In particular, the measure $\tQ_{k,M}$ defined on the set $\tcA_{k,M}$ is equivalent to a discrete Gaussian measure on $\Z^{(k-1)^2}$. This measure assigns to each $\t \in \Z^{(k-1)^2}$ the mass $\frac{1}{\tZ} \exp\bigl(-\frac{1}{2\sigma_{k,M}^2} \t'Bt\bigr)$, where $\sigma_{k,M}^2 = \frac{k-1}{k-2} \rho$ and $B$ is a symmetric positive definite matrix.

Let $\delta \in (0,\frac16)$ be fixed, and define 
$$
\cA_{k,M}(\delta) = \bigl\{A \in \cA_{k,M}:  \max_{i,j} |t_{i,j}| \le \rho^{\frac12 + \delta}\bigr\},
$$
where $T = (t_{i,j}) = A-U$. We assume $\rho \ge 4$ throughout. By Lemma~\ref{lem:gauss}, it follows that for any $A \in \cA_{k,M}(\delta)$, we have 
$$
\left| \log \frac{w(A)}{w(U)} - \log \frac{\tw(A)}{\tw(U)} \right| \le \hc_k \rho^{-\frac12 + 3 \delta},
$$
where $\hc_k$ is a positive constant depending only on $k$. Thus, for $A \in \cA_{k,M}(\delta)$, we have
\begin{equation}
\tw(A) \exp(-\hc_k \rho^{-\frac12 + 3 \delta}) \le \frac{w(A)}{w(U)} \le \tw(A) \exp(\hc_k \rho^{-\frac12 + 3 \delta})
\label{w_approx}
\end{equation}
Since $-\frac12 + 3\delta < 0$, this shows that, as $M \to \infty$ (so that $\rho \to \infty$), the ratio $w(A)/w(U)$ is well-approximated by $\tw(A)$ for all $A \in \cA_{k,M}(\delta)$. From this, we will be able to deduce that, as $M \to \infty$, the contributions made to $\E[\|A-U\|^2]$ and $\tE[\|A-U\|^2]$ by matrices in $A \in \cA_{k,M}(\delta)$ are nearly the same. The next two lemmas show that the matrices outside $\cA_{k,M}(\delta)$ make vanishingly small contributions to both the expected values.


\begin{lemma}
Let $\tcB_{k,M}(\delta) = \tcA_{k,M} \setminus \cA_{k,M}(\delta)$. We have
$$
\sum_{A \in \tcB_{k,M}(\delta)} \|A-U\|^2 \tQ_{k,M}(A) \le \kappa_k \, \rho^2 \exp\left(- \frac{k-2}{2k} \rho^{2\delta}\right),
$$
where $\kappa_k$ is a constant depending only on $k$. 
\label{lem:tQ_A^c}
\end{lemma}
\begin{proof}
For any $A \in \tcA_{k,M}$, if the entries of $T = A-U$ are bounded above in magnitude by $\rho^{\frac12 + \delta}$,
then $A$ must have non-negative entries, and hence, $A \in \cA_{k,M}(\delta)$. Therefore, for any $A \in \tcB_{k,M}(\delta)$, we have $\|A-U\|^2 = \sum_{i,j} t_{i,j}^2 \ge \rho^{1+2\delta}$. Hence,
$$
\sum_{A \in \tcB_{k,M}(\delta)} \|A-U\|^2 \, \tQ_{k,M}(A)  \le \sum_{A \in \tcA_{k,M}: \|A-U\|^2 \ge \rho^{1+2\delta}} \|A-U\|^2 \, \tQ_{k,M}(A) 
$$
and the lemma follows by applying Proposition~\ref{prop:muZR} in Appendix~F with $R = \rho^{1+2\delta}$ and $\tau = \frac{1}{k}$. 
\end{proof}

\begin{lemma}
Let $\cB_{k,M}(\delta) = \cA_{k,M} \setminus \cA_{k,M}(\delta)$. There is a positive constant $c_k'$ depending only on $k$ such that for all $A \in \cB_{k,M}(\delta)$, the bound
\begin{equation}
\frac{w(A)}{w(U)} \le \exp(- c_k' \rho^{2\delta})
\label{ineq:ck'}
\end{equation}
holds for all sufficiently large $M$. Consequently, there is a positive constant $c_k''$ depending only on $k$ such that
\begin{equation}
\sum_{A \in \cB_{k,M}(\delta)} \|A-U\|^2 \, \frac{w(A)}{w(U)} \le \exp(-c_k'' \rho^{2\delta}).
\label{ineq:ck''}
\end{equation}
holds for all sufficiently large $M$.
\label{lem:Q_A^c}
\end{lemma}
\begin{proof}
Given the bound in \eqref{ineq:ck'}, the bound in \eqref{ineq:ck''} follows readily. Indeed, note that for any $A \in \cA_{k,M}$, we have $\|A-U\|^2 \le \|A\|^2 \le k^2 M^2$. Also, note that $|\cB_{k,M}(\delta)| \le |\cA_{k,M}| \le (M+1)^{k^2}$.  Therefore, 
\begin{align*}
\sum_{A \in \cB_{k,M}(\delta)} \|A-U\|^2 \, \frac{w(A)}{w(U)} & \le \ k^2 M^2 |\cB_{k,M}(\delta)| \, \exp(- c_k' \rho^{2\delta}) \\
& \le \ \exp\bigl(- c_k' \rho^{2\delta} + O(\log M)\bigr),
\end{align*}
from which \eqref{ineq:ck''} follows.

The proof of \eqref{ineq:ck'} builds on \eqref{log_wratio}. Recall that $A-U = T$, and note that $f(x)/f(\rho) \ge 1$ iff $\rho \le x \le M+1-\rho$. For all $(i,j)$, define $\hatt_{i,j} = \min\{|t_{i,j}|,\rho^{\frac12+\delta}\}$. Then, for $t_{i,j} < 0$, we have
\begin{equation}
\sum_{\ell = -|t_{i,j}|+1}^{0} \log \frac{f(u_{i,j}+\ell)}{f(\rho)}  
\le  \sum_{\ell=-\hatt_{i,j}+1}^{0} \log \frac{f(u_{i,j}+\ell)}{f(\rho)}. 
\label{ineq:t<0}
\end{equation}
Also, for any $t_{i,j} > 0$ such that $a_{i,j} = u_{i,j} + t_{i,j} \le M+1-\rho$, we have
\begin{equation}
\sum_{\ell = 1}^{t_{i,j}} \log \frac{f(u_{i,j}+\ell)}{f(\rho)} 
\ge \sum_{\ell = 1}^{\hatt_{i,j}} \log \frac{f(u_{i,j}+\ell)}{f(\rho)}.
\label{ineq:t>0}
\end{equation}

Now, consider $A \in \cB_{k,M}(\delta)$. Suppose first that $a_{i,j} \le M+1-\rho$ for all $(i,j)$. Then, from \eqref{log_wratio}, \eqref{ineq:t<0} and \eqref{ineq:t>0}, we have
\begin{equation}
\log \frac{w(A)}{w(U)} 
   \le \sum_{(i,j):t_{i,j} < 0} \sum_{\ell = -\hatt_{i,j}+1}^{0} \log \frac{f(u_{i,j}+\ell)}{f(\rho)}  
      - \sum_{(i,j):t_{i,j} > 0} \sum_{\ell = 1}^{\hatt_{i,j}} \log \frac{f(u_{i,j}+\ell)}{f(\rho)}.
\label{bnd:log_wratio}
\end{equation}
Now arguing as in the proof of Lemma~\ref{lem:gauss} in Appendix~D (in particular, using \eqref{neg_t_bound} and \eqref{pos_t_bound}), we can bound the right-hand side above by
\begin{equation}
-\frac12 \biggl(\frac{k-2}{k-1}\biggr) \frac{1}{\rho} \sum_{i,j} (\hatt_{i,j})^2 + \frac{4}{\rho^2} \sum_{i,j} (\hatt_{i,j}+1)^3 + \frac{3}{2\rho} \sum_{i,j}\hatt_{i,j}. 
\label{bd1}
\end{equation}
Now, since $A \in \cB_{k,M}(\delta)$, we have $\sum_{i,j} (\hatt_{i,j})^2 \ge \rho^{1+2\delta}$, so that the first term above is upper bounded by $-\frac12(\frac{k-2}{k-1}) \rho^{2\delta}$. Using $\hatt_{i,j} \le \rho^{\frac12+\delta}$, the remaining two terms are upper bounded by $\hc_k \rho^{-\frac12 + 3\delta}$ for some positive constant $\hc_k$ that depends only on $k$. Hence, we have
$$
\log \frac{w(A)}{w(U)} \le -\frac12\biggl(\frac{k-2}{k-1}\biggr) \rho^{2\delta} + \hc_k \rho^{-\frac12 + 3\delta}.
$$
This proves \eqref{ineq:ck'} for $A \in \cB_{k,M}(\delta)$ with $\max_{i,j} a_{i,j} \le M+1-\rho$.

It remains to consider the case of $A \in \cB_{k,M}(\delta)$ with $\max_{i,j} a_{i,j} > M+1-\rho$. The problem here is that if $a_{i,j} > M+1-\rho$, then \eqref{ineq:t>0} may not hold, so we are unable to use the same approach as above to get to \eqref{ineq:ck'}. However, what we do now is to show that for each such $A$, there exists an $\breve{A} \in \cB_{k,M}(\delta)$ with $\max_{i,j} \breve{a}_{i,j} \le M+1-\rho$ such that $w(A) \le w(\breve{A})$. As argued above, \eqref{ineq:ck'} holds for $\breve{A}$; therefore, it holds for $A$ as well.

So, let us now prove the existence of an $\breve{A}$ as required. Let $(i,j)$ be such that $a_{i,j} > M+1-\rho$. Then, since $A \in \cA_{k,M}$, the following must hold: (i)~for all $i' \ne i$ and $j' \ne j$, we have $a_{i',j} < \rho-1$ and $a_{i,j'} < \rho-1$; and (ii)~there exists some $i'\ne i$ and $j' \ne j$ such that $a_{i',j'} > \rho$. Now, consider the matrix $A^{\pm}$ which has the same entries as $A$, except for the following: $a^{\pm}_{i,j} = a_{i,j} - 1$, $a^{\pm}_{i',j} = a_{i',j}+1$, $a^{\pm}_{i,j'} = a_{i,j'}+1$ and $a^{\pm}_{i',j'} = a_{i',j'}-1$. Note that $A^{\pm}$ is also in $\cB_{k,M}(\delta)$. 

Let $\a_1,\ldots,\a_k$ and $\a^{\pm}_1,\ldots,\a^{\pm}_k$ denote the rows of $A$ and $A^{\pm}$, respectively. Clearly, $\a_{\ell} = \a^{\pm}_{\ell}$ for all $\ell \notin \{i,i'\}$. Moreover, it can be directly verified using the definition of the function $\phi$ in \eqref{eq:phi} that $\phi(\a_{\ell}) \le \phi(\a^{\pm}_{\ell})$ for $\ell \in \{i,i'\}$. With this, we have
$$
w(A) = \prod_{\ell=1}^k \phi(\a_{\ell}) \le \prod_{\ell=1}^k \phi(\a^{\pm}_{\ell}) = w(A^{\pm}).
$$

Note that the procedure of obtaining $A^{\pm}$ from $A$ strictly reduces the $(i,j)$th entry of $A$, and does not create any new entries larger than $M+1-\rho$. If $A^{\pm}$ still contains an entry larger than $M+1-\rho$, we apply the procedure to $A^{\pm}$ to produce a matrix ${(A^{\pm})}^{\pm}$, and so on. Carrying on in this manner, after finitely many steps, we will obtain the desired matrix $\breve{A}$.
\end{proof}

\medskip

We are now in a position to complete the proof of Lemma~\ref{lem:E_and_tE}. First, we write
$$
\sum_{A \in \cA_{k,M}} \|A-U\|^2 \, \frac{w(A)}{w(U)} \ = \ 
   \sum_{A \in \cA_{k,M}(\delta)} \|A-U\|^2\, \frac{w(A)}{w(U)} + \sum_{A \in \cB_{k,M}(\delta)} \|A-U\|^2\, \frac{w(A)}{w(U)},
$$
where $\cB_{k,M}(\delta)$ is as defined in Lemma~\ref{lem:Q_A^c}. It then follows from \eqref{w_approx} and \eqref{ineq:ck''} that
there exists a positive constant $c_{1,k}$ depending only on $k$ such that, for all sufficiently large $M$,
\begin{align}
\sum_{A \in \cA_{k,M}} \|A-U\|^2 \, \frac{w(A)}{w(U)} 
  & \le \exp(c_{1,k} \, \rho^{-\frac12 + 3\delta}) \sum_{A \in \cA_{k,M}(\delta)} \|A-U\|^2 \tw(A) \notag \\
  & \le \exp(c_{1,k} \, \rho^{-\frac12 + 3\delta}) \sum_{A \in \tcA_{k,M}} \|A-U\|^2 \tw(A) \label{eq:c1} 
\end{align}

On the other hand, via \eqref{w_approx} and Lemma~\ref{lem:tQ_A^c}, we also have for all sufficiently large $M$,
\begin{align}
\sum_{A \in \cA_{k,M}} \|A-U\|^2 \, \frac{w(A)}{w(U)} 
 & \ge \sum_{A \in \cA_{k,M}(\delta)} \|A-U\|^2 \, \frac{w(A)}{w(U)} \notag \\
 & \ge \exp(-\hat{c}_k \, \rho^{-\frac12 + 3\delta}) \sum_{A \in \cA_{k,M}(\delta)} \|A-U\|^2 \, \tw(A) \notag \\
 & \ge \exp(-c_{2,k} \, \rho^{-\frac12 + 3 \delta}) \sum_{A \in \tcA_{k,M}} \|A-U\|^2 \, \tw(A) \label{eq:c2} 
\end{align}
where $c_{2,k}$ is a positive constant depending only on $k$.

Similar arguments also yield the inequalities
\begin{equation}
\exp(-c_{2,k} \, \rho^{-\frac12 + 3 \delta}) \sum_{A \in \tcA_{k,M}} \tw(A)  \le \sum_{A \in \cA_{k,M}} \frac{w(A)}{w(U)} \le \exp(c_{1,k} \, \rho^{-\frac12 + 3 \delta}) \sum_{A \in \tcA_{k,M}} \tw(A) 
\label{eq:c1c2}
\end{equation}
valid for all sufficiently large $M$.

Now, note that 
$$
\E[\|A-U\|^2] = \frac{\sum_{A \in \cA_{k,M}} \|A-U\|^2 \, \frac{w(A)}{w(U)}}{\sum_{A \in \cA_{k,M}} \frac{w(A)}{w(U)}}
$$
and 
$$
\tE[\|A-U\|^2] = \frac{\sum_{A \in \tcA_{k,M}} \|A-U\|^2 \, \tw(A)} {\sum_{A \in \tcA_{k,M}} \tw(A)}.
$$
Therefore, from \eqref{eq:c1}--\eqref{eq:c1c2}, we deduce that, with $c_{3,k} = c_{1,k} + c_{2,k}$,
$$
\exp(-c_{3,k} \, \rho^{-\frac12 + 3\delta}) \, \tE[\|A-U\|^2]  \ \le \ \E[\|A-U\|^2] \ \le \ \exp(c_{3,k} \, \rho^{-\frac12 + 3\delta}) \, \tE[\|A-U\|^2]
$$
for all sufficiently large $M$. It follows that
$$
\left|\E[\|A-U\|^2]  - \tE[\|A-U\|^2]\right| \ = \ \tE[\|A-U\|^2]| \, O\bigl(\rho^{-\frac12 + 3\delta}\bigr).
$$
Since $\tE[\|A-U\|^2] = O(\rho)$ by Lemma~\ref{lem:Qtilde_var},\footnote{Lemma~\ref{lem:Qtilde_var} is proved independently of Lemma~\ref{lem:E_and_tE}.} we conclude that 
$$
\left|\E[\|A-U\|^2]  - \tE[\|A-U\|^2]\right| = O\bigl(\rho^{\frac12 + 3\delta}\bigr),
$$
which proves Lemma~\ref{lem:E_and_tE}.


\section*{Appendix~F: Some Properties of Discrete Gaussian Measures}

In this appendix, we consider a discrete Gaussian measure defined by $\mu(\x) = \frac{1}{Z} v(\x)$, where $v(\x):=\exp\{-\frac{1}{2\beta} \x' \textrm{V} \x\}$ for $\x\in \Z^d$, $\beta > 0$ and $\textrm{V}$ a symmetric positive definite matrix, and $Z=\sum_{\x\in \Z^{d}}v(\x)$. Let $\mathbf{X}$ be a random variable distributed according to the measure $\mu$. We collect here some results on the measure $\mu$ that are used in this paper. These results are valid in the regime where $V$ is fixed and $\beta \to \infty$. 


The main tool used in the proofs in this appendix is the Poisson summation formula (see e.g., \cite[Chapter~VII, Corollary 2.6]{SW71} or \cite[Section~17]{Barvinok}). This formula applies to functions $f:\R^d \to \C$, with Fourier transform $\hat{f}$ defined for all $\xibf \in \R^d$ as $\hat{f}(\xibf) = \int_{\R^d} f(\x) e^{i \langle \xibf,\x \rangle} d\x$, such that 
$$
|f(\x)|, |\hat{f}(\x)| \le \frac{C}{1+\|\x\|^{d+\delta}} \ \ \text{ for all } \x \in \R^d
$$
for some constants $C > 0$ and $\delta > 0$. For such functions $f$, the Poisson summation formula states that
\begin{equation}
\sum_{\x \in \Z^d} f(\x) = \sum_{\xibf \in \Z^d} \hat{f}(2\pi \xibf). \label{eq:psf}
\end{equation}

By a basic fact about the Gaussian density, the function $v(\x) = \exp\{-\frac{1}{2\beta} \x^{T}\textrm{V}\x\}$ has Fourier transform $\hat{v}(\xibf) = \frac{(2\pi)^{d/2}\beta^{1/2}}{\sqrt{\det \textrm{V}}} \exp(-\frac{1}{2} \beta \, \xibf' \textrm{V}^{-1} \xibf)$. Hence, the Poisson summation formula applies, and we have
\begin{equation}
Z = \sum_{\x\in \Z^{d}}v(\x) = \sum_{\xibf\in\Z^d} \hat{v}(2\pi\xibf) 
= \frac{(2\pi)^{d/2} \beta^{1/2}}{\sqrt{\det \textrm{V}}} \, Z^*
\label{eq:Z}
\end{equation}
where we define $Z^* =  \sum_{\xibf\in\Z^d} \exp(-\frac{1}{2} 4\pi^2 \beta \, \xibf' \textrm{V}^{-1} \xibf)$.

Clearly, $Z^* \ge 1$ since the $\xibf = \0$ term in the sum evaluates to $1$. In fact, as $\beta \to \infty$, $Z^* \to 1$. This is because $\xibf' V^{-1} \xibf$ is a positive definite quadratic form, so that $\lim_{\beta \to \infty} \exp(-\frac{1}{2} 4\pi^2 \beta \, \xibf' \textrm{V}^{-1} \xibf)$ equals $0$ if $\xibf \ne \0$, and equals $1$ if $\xibf = \0$. Thus, $Z \to \frac{(2\pi)^{d/2} \beta^{1/2}}{\sqrt{\det \textrm{V}}}$. To estimate the rate of this convergence, we make use of some bounds on the quadratic form $\xibf' V^{-1} \xibf$. 

The matrix $\textrm{V}^{-1}$ can be diagonalized as $U'\Lambda^{-1}U$, where $U$ is an orthogonal matrix and $\Lambda = \text{diag}(\lam_1,\ldots,\lam_d)$ is a diagonal matrix composed of the eigenvalues, $\lam_1,\ldots,\lam_d$, of $\textrm{V}$. Since $\textrm{V}$ is symmetric and positive definite, its eigenvalues are all real and positive. Thus, with $\etabf = U\xibf$, we have $\xibf' \textrm{V}^{-1} \xibf = \etabf' \Lambda^{-1} \etabf = \sum_{i=1}^d \frac{1}{\lam_i} \eta_i^2$. Hence, letting $\lambda_{\min}$ and $\lambda_{\max}$ denote the smallest and largest eigenvalues, respectively, of $\textrm{V}$, we have
\begin{equation}
\frac{1}{\lambda_{\max}} \|\xibf\|^2 = \frac{1}{\lambda_{\max}} \|\etabf\|^2 \le \xibf' \textrm{V}^{-1} \xibf \le \frac{1}{\lambda_{\min}} \|\etabf\|^2 = \frac{1}{\lambda_{\min}} \|\xibf\|^2,
\label{qform_bounds}
\end{equation}
the equalities on either side being due to the fact that orthogonal transformations preserve $\ell_2$ norms.

\begin{proposition}
In the regime where $\beta \to \infty$, we have $Z = \frac{(2\pi)^{d/2} \beta^{1/2}}{\sqrt{\det \textrm{V}}}\left[1 + O\bigl(\exp(-\frac{2\pi^2}{\lam_{\max}} \, \beta)\bigr) \right]$.
\label{prop:Zgauss}
\end{proposition}
\begin{proof}
Let $C_0 = \frac{2\pi^2}{\lam_{\max}}$. It is enough to show that $Z^* = 1 + O\bigl(\exp(-C_0 \beta)\bigr)$. Since $Z^* \ge 1$, we need a corresponding upper bound. This is done using \eqref{qform_bounds} as follows:
\begin{align*}
Z^* \  \le \ \sum_{\xibf\in\Z^d} \exp\biggl(-C_0 \, \beta \sum_{i=1}^d \xi_i^2\biggr) 
& = \ \prod_{i=1}^d \sum_{\xi_i \in \Z} \exp\biggl(-C_0 \beta \, \xi_i^2\biggr) \\
& = \biggl[\sum_{\xi \in \Z} \exp\biggl(-C_0 \beta \, \xi^2\biggr) \biggr]^d 
\ \le \ \biggl[\sum_{\xi \in \Z} \exp\biggl(-C_0 \beta \, |\xi|\biggr) \biggr]^d 
\end{align*}
The upper bound $1 + O\bigl(\exp(-C_0 \beta)\bigr)$ now follows from the geometric series summation formula.
\end{proof}

The Poisson summation formula also applies to the function $u(\x) = \x'\textrm{V}\x \, v(\x)$. Recall that if $f$ has a twice-differentiable Fourier transform, then the function $g(\x) = \|\x\|^2 f(\x)$ has Fourier transform $\hat{g}(\xibf) = - \Delta\hat{f}(\xibf)$, where $\Delta\hat{f}$ is the Laplacian of $\hat{f}$. With $f(\x) = \exp(-\frac{1}{2} \|\x\|^2)$, we have $g(\x) = \|\x\|^2 \exp(-\frac{1}{2} \|\x\|^2)$, and $\frac{1}{\beta} u(\x) = g(\beta^{-1/2} \textrm{V}^{1/2} \x)$, where $\textrm{V}^{1/2}$ is the symmetric positive definite square root of $\textrm{V}$. We know that $\hat{f}(\xibf) = (2\pi)^{d/2} \exp(-\frac{1}{2} \|\xibf\|^2)$, from which straightforward computations yield $\hat{g}(\xibf) = (2\pi)^{d/2} (d - \|\xibf\|^2) \, \exp(-\frac12\|\xibf\|^2)$. Now, via a change of variable, 
\begin{align}
\frac{1}{\beta} \hat{u}(\xibf) 
& = \frac{(2\pi)^{d/2} \beta^{1/2}}{\sqrt{\det \textrm{V}}} (d - \beta \xibf'\textrm{V}^{-1}\xibf) \, \exp(-\frac12 \beta \, \xibf'\textrm{V}^{-1}\xibf) \notag \\
& = \frac{Z}{Z^*} (d - \beta \, \xibf'\textrm{V}^{-1}\xibf) \, \exp(-\frac12 \beta \, \xibf'\textrm{V}^{-1}\xibf),\label{eq:uhat}
\end{align}
where we have used the identity in \eqref{eq:Z}. 

\begin{proposition}
In the regime where $\beta \to \infty$, we have
$$
\E_{\mu}[\mathbf{X}'\mbox{\emph V} \mathbf{X}] = \frac{1}{Z} \sum_{\x\in\Z^d} u(\x) = \beta \, d + O\left(\beta^2 \exp\biggl(-\frac{2\pi^2}{\lam_{\max}} \, \beta\biggr)\right).
$$
\label{prop:var}
\end{proposition}

\begin{proof}
Plugging \eqref{eq:uhat} into the Poisson summation formula \eqref{eq:psf}, we obtain
\begin{align}
\E_\mu[\mathbf{X}'\textrm{V}\mathbf{X}] \ = \ \frac{1}{Z} \sum_{\x \in \Z^d} u(\x) 
& = \frac{1}{Z} \sum_{\xibf \in \Z^d} \hat{u}(2\pi\xibf) \notag \\
& =  \frac{\beta}{Z^*} \sum_{\xibf \in \Z^d} (d - 4\pi^2 \beta \, \xibf'\textrm{V}^{-1}\xibf) \, \exp\bigl(-\frac12 4\pi^2 \beta \, \xibf'\textrm{V}^{-1}\xibf\bigr) \notag \\
& = \beta \, d \ - \ \frac{4\pi^2\beta^2}{Z^*} \sum_{\xibf \in \Z^d} \xibf'\textrm{V}^{-1}\xibf \, \exp\bigl(-\frac12 4\pi^2 \beta \, \xibf'\textrm{V}^{-1}\xibf\bigr).
\label{eq:ExQx}
\end{align}

It only remains to show that $\sum_{\xibf \in \Z^d} \xibf'\textrm{V}^{-1}\xibf \, \exp\bigl(-\frac12 4\pi^2 \beta \, \xibf'\textrm{V}^{-1}\xibf\bigr)$ is $O\bigl(\exp(-\frac{2\pi^2}{\lam_{\max}} \, \beta)\bigr)$. Using \eqref{qform_bounds}, the summand can be bounded as 
$$
 \xibf'\textrm{V}^{-1}\xibf \, \exp\biggl(-\frac12 4\pi^2 \beta \, \xibf'\textrm{V}^{-1}\xibf\biggr) 
   \le  \frac{1}{\lambda_{\min}} \|\xibf\|^2 \exp\biggl(\frac{-2\pi^2}{\lambda_{\max}} \, \beta \|\xibf\|^2\biggr). 
$$
We then have
\begin{align}
 \sum_{\xibf \in \Z^d} \xibf'\textrm{V}^{-1}\xibf \, & \exp\biggl(-\frac12 4\pi^2 \beta \, \xibf'\textrm{V}^{-1}\xibf\biggr) \notag \\
& \le \frac{1}{\lambda_{\min}} \, \sum_{\xibf \in \Z^d} \sum_{i=1}^d \xi_i^2 \exp\biggl(\frac{-2\pi^2}{\lambda_{\max}} \, \beta \sum_{i=1}^d\xi_i^2\biggr) \notag \\
& = \frac{1}{\lambda_{\min}} \, \sum_{i=1}^d \sum_{\xibf \in \Z^d} \xi_i^2 \exp\biggl(\frac{-2\pi^2}{\lambda_{\max}} \, \beta \sum_{i=1}^d\xi_i^2\biggr) \notag \\
& =  \frac{d}{\lambda_{\min}} \, \sum_{\xibf \in \Z^d} \xi_1^2 \exp\biggl(\frac{-2\pi^2}{\lambda_{\max}} \, \beta \sum_{i=1}^d\xi_i^2\biggr) \notag \\
& =  \frac{d}{\lambda_{\min}} \, \left[\sum_{\xi_1\in \Z} \xi_1^2 \exp\biggl(\frac{-2\pi^2}{\lambda_{\max}} \, \beta \, \xi_1^2\biggr)\right] \, \prod_{i=2}^d \left[ \sum_{\xi_i\in \Z} \exp\biggl(\frac{-2\pi^2}{\lambda_{\max}} \, \beta \, \xi_i^2\biggr) \right]\notag \\
& =  \frac{d}{\lambda_{\min}} \, \left[\sum_{\xi\in \Z} \xi^2 \exp\biggl(\frac{-2\pi^2}{\lambda_{\max}} \, \beta \, \xi^2\biggr)\right] \, {\left[ \sum_{\xi\in \Z} \exp\biggl(\frac{-2\pi^2}{\lambda_{\max}} \, \beta \, \xi^2\biggr) \right]}^{d-1} \notag \\
& \le \frac{d}{\lambda_{\min}} \, \left[\sum_{\xi\in \Z} \xi^2 \exp\biggl(\frac{-2\pi^2}{\lambda_{\max}} \, \beta \, |\xi|\biggr)\right] \, {\left[ \sum_{\xi\in \Z} \exp\biggl(\frac{-2\pi^2}{\lambda_{\max}} \, \beta \, |\xi|\biggr) \right]}^{d-1} \notag \\
& = \frac{d}{\lambda_{\min}} \, O\biggl(\exp\bigl(-\frac{2\pi^2}{\lam_{\max}} \, \beta\bigr)\biggr) 
  {\left[1 + O\biggl(\exp\bigl(-\frac{2\pi^2}{\lam_{\max}} \, \beta\bigr)\biggr) \right]}^{d-1} \label{eq:O} \\
& = O\biggl(\exp\bigl(-\frac{2\pi^2}{\lam_{\max}} \, \beta\bigr)\biggr) \notag
\end{align}
the equality in \eqref{eq:O} being a consequence of standard geometric series summation formulas.
\end{proof}

Our final result estimates the contribution to 
$\E_\mu[\X'V\X] = \frac{1}{Z} \sum_{\x} u(\x)$ made by vectors $\x \in \Z_d$ with $\x'V'\x \ge R$ for some (large) $R > 0$. To this end, define $\mfZ(R) := \{\x \in \Z^d: \x'V\x \ge R\}$. 

\begin{proposition} For any $R > 0$ and $0 < \tau < 1$, we have 
%
$$
\frac{1}{Z} \sum_{\x \in \mfZ(R)} u(\x) \le \beta d \, \tau^{-(\frac{d}{2}+1)} \exp\left(- \frac{(1-\tau)R}{2\beta}\right).
$$
\label{prop:muZR}
\end{proposition}
\begin{proof}
For $0 < \tau < 1$, we write
\begin{align}
\frac{1}{Z} \sum_{\x \in \mfZ(R)} u(\x) 
& = \ \frac{1}{Z} \sum_{\x \in \mfZ(R)} \x'V\x \, \exp\biggl(-\frac{1-\tau}{2\beta} \x'V\x\biggr)  \exp\biggl(-\frac{\tau}{2\beta} \x'V\x\biggr) \notag \\
& \le \ \exp\biggl(-\frac{1-\tau}{2\beta} R \biggr) \, \frac{1}{Z} \sum_{\x \in \Z^d} \x'V\x \, \exp\biggl(-\frac{\tau}{2\beta} \x'V\x \biggr) \notag \\
& = \ \exp\biggl(-\frac{1-\tau}{2\beta} R \biggr) \, \tau^{-1} \, \frac{1}{Z} \sum_{\x \in \Z^d} u(\tau^{\frac12}\x) \notag \\
& = \ \exp\biggl(-\frac{1-\tau}{2\beta} R \biggr) \, \tau^{-(\frac{d}{2}+1)} \, \frac{1}{Z} \sum_{\xibf \in \Z^d} \hat{u}(2\pi\tau^{-\frac12}\xibf)
\label{eq:tau3} \\
& \le \exp\biggl(-\frac{1-\tau}{2\beta} R \biggr) \, \tau^{-(\frac{d}{2}+1)} \, \frac{\beta}{Z^*} \sum_{\xibf \in \Z^d} d \exp(-\frac12 4\pi^2\tau^{-1} \beta \, \xibf'\textrm{V}^{-1}\xibf) \label{ineq:tau2} \\
& \le \exp\biggl(-\frac{1-\tau}{2\beta} R \biggr) \, \tau^{-(\frac{d}{2}+1)} \, \frac{\beta}{Z^*} \sum_{\xibf \in \Z^d} d \exp(-\frac12 4\pi^2 \beta \, \xibf'\textrm{V}^{-1}\xibf) \notag \\
& = \exp\biggl(-\frac{1-\tau}{2\beta} R \biggr) \, \tau^{-(\frac{d}{2}+1)} \, \beta d. \notag
\end{align}
In \eqref{eq:tau3} above, we used the Poisson summation formula \eqref{eq:psf}, and \eqref{ineq:tau2} follows from \eqref{eq:uhat}.
\end{proof}

\section*{Acknowledgement}
We would like to thank Pascal Vontobel for suggesting the approach of analyzing the Bethe approximation via degree-$M$ lifted permanents.


\end{document}